\documentclass[11pt]{amsart}
\usepackage{amsbsy,amssymb,amscd,amsfonts,latexsym,amstext,delarray, amsmath,graphicx,color,caption}
\usepackage{tikz,lscape}
\usetikzlibrary{matrix,arrows}
\usepackage{graphicx}
\input xy

\xyoption{all}
\usepackage{enumerate} 

\newcommand{\R}{\mathbb{R}}

\newcommand{\C}{\mathbb{C}}

\newcommand{\Z}{\mathbb{Z}}
\newcommand{\SL}{\operatorname{SL}}

\newcommand{\GL}{\operatorname{GL}}

\newcommand{\PSL}{\operatorname{PSL}}
\newcommand{\SO}{\operatorname{SO}}
\newcommand{\Sp}{\operatorname{Sp}}
\newcommand{\Spin}{\operatorname{Spin}}
\newcommand{\SU}{\operatorname{SU}}

\newcommand{\Id}{\operatorname{Id}}
\newcommand{\ImC}{\operatorname{Im}}
\newcommand{\ReC}{\operatorname{Re}}

\newcommand{\wt}{\widetilde{w}}

\newtheorem{theorem}{Theorem}[section]
\newtheorem{proposition}[theorem]{Proposition}
\newtheorem{lemma}[theorem]{Lemma}

\def\Aut{\operatorname{Aut}}

\def\a{\alpha}

\def\endofproof {\hfill{$\Box$}\\ }

\renewcommand{\(}{\begin{equation}}
\renewcommand{\)}{\end{equation}}
\newcommand{\bea}{\begin{eqnarray}}
\newcommand{\eea}{\end{eqnarray}}
\begin{document}

\title{\bf Integral group actions on symmetric spaces and \\ discrete duality symmetries of supergravity theories}

\date{\today}

\author{Lisa Carbone, Scott H.\ Murray and Hisham Sati}

\thanks{The first author was supported in part by NSF grant \#DMS-1101282, and the third author 
was supported by NSF grant \#PHY-1102218 and a University of Pittsburgh CRDF grant.}

\begin{abstract} For $G=G(\mathbb{R})$ a split, simply connected, semisimple Lie group of rank $n$ and 
$K$ the maximal compact subgroup of $G$, we give a method for computing Iwasawa coordinates of $K\backslash G$ using the Chevalley generators and the Steinberg presentation.  When $K\backslash G$ is a scalar coset for a supergravity theory in dimensions $\geq 3$, we determine the action of the integral form $G(\Z)$  on $K\backslash G$. We give explicit results for the action of the discrete U-duality groups $\SL_2(\mathbb{Z})$ and $E_7(\mathbb{Z})$ on the scalar cosets $\SO(2)\backslash \SL_2(\mathbb{R})$  and $[\SU(8)/\{\pm \Id\}]\backslash E_{7(+7)}(\mathbb{R})$  for  type IIB supergravity in   ten dimensions and 11-dimensional supergravity reduced to $D=4$ dimensions, respectively. For the former, we use this to determine the discrete U-duality transformations on the scalar sector in the Borel gauge and we describe the discrete symmetries of the dyonic charge lattice. We  determine the spectrum-generating symmetry group for fundamental BPS solitons of type IIB supergravity in $D=10$ dimensions at the classical level and we propose an analog of this symmetry at the quantum level. We indicate how our methods can be used to study the orbits of discrete U-duality groups in general.
 
\end{abstract}

\maketitle

\tableofcontents

\section{Introduction}

In classical supergravity theories, one is interested in classifying 
solutions of the equations of motion that enjoy certain desirable properties. 
These often take the form of extremal black $p$-brane 
solutions, and are usually related by the underlying algebraic structure of classical U-duality. 
In the quantum theory, that is in string theory or in M-theory, 
the corresponding solutions are believed to be related by 
discrete U-duality (\cite{HT}), leading to a subtle and complex orbit structure under the discrete U-duality group.

Thus knowing the orbits of the U-duality group  on the space of solutions
allows one, at least in principle, 
to know all solutions of a given class of equations in the quantum theory. 
Such solutions are usually charged under (higher analogs of) electric and magnetic fields. 
The electric and magnetic charge 
vectors of the asymptotically flat 
$p$-brane solutions form irreducible 
representations of the U-duality 
group (\cite{LPS}). 

The question of determining the distinct charge vector orbits under  U-duality 
is of central importance. The U-duality orbits for real valued charges are classified and have been understood for some time (\cite{FG}, \cite{LPS}).  However, a classification of orbits under the discrete U-duality group in known only in certain cases (see  [BDDFMR] for an excellent overview of the status of this question and the open issues).

We recall that in dimensionally reduced supergravity theories, there is a Lie group $G$ and subgroup $K$ such that the scalar fields take values in the space $K\backslash G$
of right cosets. 
This coset space has a $K$-action on the left and a right action by the discrete U-duality group $G(\Z)$. In most cases, $G=G(\mathbb{R})$ is a split, simply connected, semisimple algebraic group of rank $n$ and $K$ is its maximal compact subgroup.

In this work,  we give a method using the Steinberg presentation for computing Iwasawa coordinates of the quotient $K\backslash G$ and the action of the $\mathbb{Z}$-form $G(\mathbb{Z})$ on $K\backslash G$  by means of Iwasawa decomposition $G=KH^+U$ (see Sections~\ref{S-action-sl2} and~\ref{Sec Iwasawa}). 
These computations encode the discrete U-duality symmetries on the scalar coset $K\,\backslash G$. Similar computations are common in the theory of automorphic forms on finite dimensional Lie groups (see \cite{H-C}). 

Our computations hence involve finding the standard Iwasawa form of a generic element $g\in G$ and then allowing $z\in G(\mathbb{Z})$ to act on the coset $Kg$. 
The task is then to determine a compensating element $k\in K$ so that $kgz$ is again in standard Iwasawa form.
This requires completely determining the constants given in the Steinberg relations for the  group $G$ (\cite{St})
and explicitly describing $K$ and $G(\Z)$ in terms of the Steinberg generators. The computer algebra system Magma (\cite{BCFS}) is essential here, since it automatically computes these constants without infeasible hand calculations. 

The hardest part of constructing the action of $G(\Z)$ seems to be the action of  the Weyl group representatives in $G$.
Tables~1 through~8 give the Iwasawa coordinates and compensating elements for the action of these elements, for all finite-dimensional 
U-duality groups $G$.
As examples,  we compute  the action of the generators of the discrete groups
 $\SL_2(\mathbb{Z})$ and $E_{7(+7)}(\mathbb{Z})$ on the cosets  $\SO(2)\backslash \SL_2(\mathbb{R})$ and  $[\SU(8)/\{\pm \Id\}]\backslash E_{7(+7)}(\mathbb{R})$. These cosets are known to encode the scalar fields of  type IIB supergravity in $D=10$ dimensions and 11-dimensional supergravity reduced to  $D=4$ 
 dimensions, respectively. 

These techniques can also be extended to determine the orbits of $G(\Z)$ on the charge lattice. 
Our methods are general and do not depend on the choice of group $G$. Our techniques could hence be applied to all U-duality groups, including the hyperbolic Kac--Moody group $E_{10}$, which is conjectured to be the U-duality group of 11-dimensional supergravity reduced to 1 dimension (see Section \ref{other}).

Earlier fundamental works, particularly \cite{BDDFMR},  make use of Jordan algebras and 
Freudenthal systems (as described in \cite{Bh} and \cite{Kr}) and of integral forms of groups
(as described in \cite{BDDR}, \cite{EG} and \cite{Gr}) to study extremal black $p$-brane 
solutions of supergravity in 6, 5 and 4-dimensional classical theories and their relationship under U-duality. 
%
%
After preparation of this manuscript, we also learned of the work of Cacciatori, Cerchiai and Marrani (\cite{CC}, \cite{CCM1}--\cite{CCM5}) who
studied various properties of the Iwasawa parametrization and gave a general method for constructing it explicitly at the group level.
In particular, in \cite{CCM1} and \cite{CCM4},  they gave Iwasawa coordinates in the  symplectic frame for the symmetric space $[\SU(8)/\{\pm \Id\}]\backslash E_{7(+7)}(\mathbb{R})$  and described some of the implications for the corresponding supergravity theory.

We are grateful to Thibault Damour for suggesting the study of this problem and for helpful discussions. We also wish to thank Arun Ram for providing his unpublished results that we used in Section 6.

\section{Preliminaries on U-duality }

A rank $d$ U-duality group 
arises in the reduction of M-theory on a $(d+1)$-torus $T^{d+1}$.
It is generated by the T-duality group $\SO(d, d)$ and the S-duality group
$\SL_{d+1}(\Z)$, which acts geometrically as the modular group of the torus $T^{d+1}$. 
The U-duality group is then 
$$ \SL_{d+1}(\Z) \bowtie \SO_{d, d}( \Z),$$
the group generated by the two non-commuting 
actions  (see \cite{OP} for a survey). 

For example,  dimensional reduction of M-theory on the torus $T^7$ leads to $D=4$, $N=8$  maximal 
supergravity which has non-compact $E_{7(+7)}(\R)$ global symmetry action on the 
56-dimensional fundamental representation of the Lie algebra $E_7$. 

In $D=4$ spacetime dimensions, 11-dimensional supergravity has the  maximal  number $N=8$ of supersymmetries. The equations of motion and the Bianchi identities of 
the $N=8$ supergravity theory in  four dimensions are invariant  under the non-compact U-duality group $E_{7(+7)}(\R)$ (\cite{CJ}).

In $D=4$ dimensions, the electric and magnetic charges are subject to the Dirac--Zwanziger--Schwinger quantization condition. 
The electric and magnetic charges live on a lattice $\mathcal{Q}$ in a 56-dimensional vector space $V$, with 28 electric and 28 magnetic fundamental charges. 
The  representation $V$ gives rise to a faithful representation of $E_{7(+7)}(\R)$ in $\Sp_{56}(\mathbb{R})$. 

The discrete group $E_{7(+7)}(\mathbb{Z})$ acts on (\cite{CJ}):
\begin{itemize}
\item the set of magnitudes of electric and magnetic charges on the lattice $\mathcal{Q}$ given by the Dirac--Zwanziger--Schwinger quantization condition;
\item the abelian gauge fields as generalized electromagnetic duality; and
\item the 70 scalar fields of the theory, which take values in the coset space 
$$[\SU(8)/\{\pm \Id\}]\backslash E_{7(+7)}(\mathbb{R}).$$
\end{itemize}

 The $\Z$-form $G(\Z)$, for any simple and simply connected Chevalley group or Kac--Moody group $G$, may be defined as the stabilizer 
$$G({\mathbb{Z}})=\{g\in G({\mathbb{R}})\mid g\cdot V_{\Z}\subseteq V_{\Z}\}$$
of a lattice $V_{\Z}$ in a fundamental representation $V$ for $G$ ([St], [BC]).

So for $G=E_7$ we get
$$E_{7(+7)}(\mathbb{Z}) = E_{7(+7)}(\mathbb{R}) \cap \Sp_{56}(\mathbb{Z}),$$
as discovered in [HT], following [CJ] in the framework of type II string theory. 

Soul\'e (\cite{S}) gave a rigorous mathematical proof that the $E_{7(+7)}(\mathbb{Z})$ of Hull and Townsend coincides with the Chevalley $\Z$-form of $G=E_7$. 
Here $E_{7(+7)}(\mathbb{R}) \cap \Sp_{56}(\mathbb{Z})$ is the stabilizer of the standard lattice in the 56-dimensional fundamental representation of $E_7$. (See also\cite{MS}  for a discussion of the role of $E_8$).

\subsection{$\SL_2(\Z)$-symmetry in supergravity}
 The U-duality group $G(\Z)= E_{n(+n)}(\Z)$
has  many subgroups isomorphic to $\SL_2(\Z)$. We highlight two particularly distinguished ones.
The first  interchanges the NS-NS (Neveu--Schwarz) fields with the R-R 
(Ramond-Ramond) fields
(called X-duality in \cite{LPS1})
and is conjectured to be a
non-perturbative symmetry of 
 type IIB superstring theory
 in $D=10$ (\cite{Sc}) and  type IIA superstring theory in 
$D=9$ (\cite{BHO}). 

The second $\SL_2(\Z)$ subgroup implements 
electric-magnetic duality (\cite{DL}), which exists
only in $D=4$ space-time dimensions, 
and is again non-perturbative.



 \subsection{The role of the Weyl group}
 The role of the Weyl group $W$ associated to  the discrete
U-duality group $G(\Z)$ is analogous to that of the $\Z_2$ subgroup of 
the U(1) electric-magnetic duality group 
in Maxwell theory, which describes the discrete interchange of electric and magnetic 
fields
$$
E \mapsto B {\rm ~and~} B \mapsto -E.$$
The Weyl group acts as rotations by integral multiples of $\pi/2$ on the 
space of field strengths. In contrast, the full U-duality group $G$ (over $\R$ or $\Z$)
includes more general rotations in the space of field strengths (\cite{LPS1}).
The Weyl group preserves the total number of 
electric and magnetic charges, whereas the U-duality group does not, so that $W$ gives
a characterization of the independent $p$-brane solutions of a given type (\cite{LPS1}).




\subsection{The classical coset}

The dimensional reduction of 11-dimensional supergravity on 
a $d$-torus gives rise to a theory in $D$ dimensions with the 
scalar fields having the following symmetry pattern given 
by Cremmer and Julia \cite{CJ1} (see \cite{HPS}, \cite{Ni}):
$$
\begin{tabular}{|c|c||l|l|l|l|}
\hline
$D$ & $d$ & $E_d$&  $G=E_{d(+d)}$ & $K$ & $\#$scalars $=\dim K\backslash G$\\
\hline
\hline
9 & 2 & $A_{1}$ &$\SL_2 (\R)$  & $\SO(2)$ & $4 -1=3$\\
8 & 3 & $A_2A_1$ &$\SL_3(\R) \times \SL_2(\R)$ & $\SO(3)\times \SO(2)$ & $11 - 4=7$\\
7 & 4 & $A_4$&$\SL_5(\R)$ &  $\SO(5)$ & $24-10=14$\\
6 & 5 & $D_5$&$\Spin(5,5)$ & $\left(\Spin(5) \times  \Spin(5)\right)/\{\pm\Id\}$ & $45 - 20=25$\\
5 & 6 & $E_6$&$E_{6(+6)}(\R)$ & $\Sp(8)/\{\pm\Id\}$ & $78 - 36=42$ \\
4 & 7 & $E_7$&$E_{7(+7)}(\R)$ & $\SU(8)/\{\pm\Id\}$ & $133 - 63=70$\\
3 & 8 & $E_8$&$E_{8(+8)}(\R)$ & $\Spin(16)/\{\pm\Id\}$ & $248 - 120=128$\\
\hline
\end{tabular}$$
We recall that the notation $E_{d(+d)}$ refers to the non-compact split, simply connected form of $E_d$.


A systematic way to study the orbits of the global symmetry groups $G=E_{\tiny {11-D},+(11-D)}$ in $D$ dimensions, on  the charge vector
space,  is to start with the simplest solution
and then apply the global symmetry group $G$. Since the global symmetry commutes
with supersymmetry, the generated solutions will also be supersymmetric. This study was carried out in \cite{LPS} for
$4 \leq D \leq 9$ where a classification of orbits was also given.

To give the spectrum of BPS solitons, one needs to classify the sets of 
solutions at fixed values of the scalar moduli, that is, the asymptotic values of all the diatonic and 
axionic scalars.


We recall the case of $D=4$, following  \cite{CJ}, \cite{HT}, \cite{ADFFT} and  \cite{BFT}. 
The equations of motion 
and the Bianchi identities of the $N=8$ classical supergravity theory in  four dimensions are invariant 
under the classical U-duality group $E_{7(+7)}(\R)$. The group $E_{7(+7)}(\R)$  
acts simultaneously on both the space spanned by the
70 scalar fields $\phi^{\alpha}$ and on the vector space {\bf\scshape{Q}} generated by the 28 electric and 28 
magnetic quantized charges
 in the 56-dimensional fundamental representation.
  
 A static, 
spherically symmetric BPS black hole solution is characterized  in general by a vector $\vec{Q}\in$ {\bf\scshape{Q}}
and a particular point $\phi_\infty$ on the moduli space of the theory whose $70$ coordinates 
$\phi^\alpha_\infty$ are the values of the scalar fields at spatial infinity ($r\rightarrow \infty$). 

Acting 
on a black hole solution ($\phi_\infty\,,\,\vec Q$) by means of a $U$-duality transformation $g\in E_{7(+7)}(\R)$ one 
generates a new black hole solution ($\phi_\infty^g\,,\,{\vec Q}^g$). The BPS black hole solutions therefore fill the U-duality orbits.

The orbits of the classical $E_{7(+7)}(\R)$ group are studied in
\cite{FM}, \cite{FG} and  \cite{LPS}. The orbits can be viewed as 
similar, in a sense, to the orbits of 
time-like, light-like and space-like vectors in Minkowski space,
except that one uses a quartic invariant $I_4$ instead of a 
quadratic form.

%

The different orbits with various supersymmetries can
also be related to intersecting branes and are
characterized by certain group invariant polynomials (\cite{FM}),
namely the analogous quartic invariant $I_4$ in four dimensions  (\cite{KK}).

\subsection{The quantum coset}
In the full quantum theory, it is conjectured that charges are quantized and the duality symmetry is 
broken to the discrete U-duality subgroup $G(\Z)$ (\cite{HT})
as a consequence of the Dirac--Zwanziger--Schwinger charge quantization 
condition.  Classical supergravity solutions correspond to the limit of large values of integer
quantized charges. 

 The orbits of the charge vectors under the integral group $G(\Z)$ are  subtle and 
a complete characterization is not yet known. 
Partial results on discrete orbit classification
are made in certain cases by introducing new arithmetic U-duality 
invariants not appearing in the real form. These are given by the
`greatest common divisor' of U-duality representations built out of the basic charge
vector representations (\cite{BDDFMR}, \cite{BDDR},  \cite{DGN}, \cite{Se}).
 
However, a classification of orbits under the discrete U-duality group in known  in certain cases. For example, for the subclass of `projective black holes' satisfying an arithmetic condition (Section 4.4 of \cite{BDDFMR}), a complete classification was given in \cite{BDDFMR} in terms of an $E_{7(+7)}(\Z)$-orbit of a certain canonical form. Here the techniques involved the use of integral Jordan algebras, the integral Freudenthal triple system and the work of Krutelevich (\cite{Kr}). In \cite{BDDFMR}, it was shown that all black holes with the same quartic norm are U-duality related, while the situation for non-projective black holes remains unclear.

%

\subsection{The relationship between $G(\R)$ and $G(\Z)$ orbits} 

For completeness, 
here we state two observations, the first from 
\cite{BDDFMR} and the second from \cite{FMMS}. 

First, the $G(\Z)$ orbits fall into disjoint sets corresponding  to the orbits of the classical theory under the continuous group $G(\R)$. This follows from the fact that
 the conditions separating the continuous orbits are manifestly
invariant under the corresponding discrete U-dualities and hence states that are 
unrelated in the continuous case remain unrelated in the discrete case.

Second, 
in order for the greatest common divisor of a collection of U-duality 
representations to be well
defined, such  representations must be non-vanishing (\cite{BDDFMR}). One must then identify the class of orbits to which a given state belongs to in the continuous case. This in turn yields an  identification of the subset of arithmetic invariants that are well defined for this
particular state.

\subsection{The impact on M-theory of identifying the $E_{7(+7)}(\Z)$ invariants.} 
 Manifestly $E_{7(+7)}(\Z)$ invariant 
 partition functions would tell us about
 the full microscopic physics of M-theory (on a  seven-torus). 

As discussed above, the integral form $E_{7(+7)}(\Z)$ has discrete invariants that are not seen in the 
real form and are given 
by the greatest common divisor of certain sets of numbers 
which correspond to covariant tensors of $E_{7(+7)}(\R)$
(\cite{BDDR}, \cite{Se}). 

As shown in \cite{BFK}, some of the orbits of $E_{7(+7)}(\Z)$ should play an important role
in counting the micro-states of $D=4$, $N=8$ supergravity.

A relation between time-like, light-like and space-like orbits of the 
$E_{7(+7)}(\R)$ symmetry and 
discrete $E_{7(+7)}(\Z)$ invariants is established in \cite{BFK}. 
The time-like, light-like, and space-like orbits in $E_{7(+7)}(\R)$ 
corresponds to $I_4>0$, $I_4=0$, and $I_4<0$, respectively.

\section{Action of $\SL_2(\mathbb{Z})$ on $\SO(2)\backslash \SL_2(\mathbb{R})$}
\label{S-action-sl2}

In this section we work out the case for $G=\SL_2(\R)$ in detail.  We begin with the details of  the action of $\SL_2(\mathbb{Z})$ on $\SO(2)\backslash \SL_2(\mathbb{R})$ in matrix form and then reduce it to a discussion that uses only Chevalley generators of  the group $\SL_2(\R)$. This   allows us to generalize the discussion to higher rank in the next section. 

Recall that $ {\SO}(2)= 
\{A\in \SL_2(\mathbb{R})\mid AA^T=A^TA=\Id\}.$
 The group $\SL_2(\mathbb{R})$ acts on  the  Poincar\'e upper half plane 
 $\mathbb{H}=\{z\in\mathbb{C}\mid \ImC(z)>0\}$ 
by fractional linear transformations
$$\left(\begin{matrix}  a & b \\ c & d  \end{matrix}\right):  z\mapsto \dfrac{az+b}{cz+d}.$$
%
%
As is well known (see \cite{Mi}), the action of  $\SL_2(\mathbb{R})$ on $\mathbb{H}$ is transitive and the stabiliser of $i\in \mathbb{H}$ is $\SO(2)$.
Hence
the coset space $\SO(2)\backslash \SL_2(\mathbb{R})$ is homeomorphic to $\mathbb{H}$ via the map  
$$\SO(2)A \mapsto \dfrac{ai+b}{ci+d}\qquad\text{where $A=\left(\begin{matrix}  a & b \\ c & d  \end{matrix}\right)$.}$$ 

The $\Z$-form  $G(\Z)=\SL_2(\mathbb{Z})$ acts on $\mathbb{H}$ by integral fractional linear transformations with fundamental domain
 $\mathcal{F}=\{z\in \mathbb{H}\mid |z|\geq 1 \text{ and } |\ReC(z)|\leq 1/2\}$.
It is well known that $\SL_2(\mathbb{Z})$ is generated by
$$T=\left(\begin{matrix} 1 & 1\\ 0 & 1\end{matrix}\right):z\mapsto z+1 \quad\text{and}\quad S=\left(\begin{matrix} 0 & 1\\ -1 & 0\end{matrix}\right):z\mapsto -1/z.$$
Writing $a=TS$, $b=S$ we obtain group presentations
\begin{align*}
\SL_2(\mathbb{Z})&=\langle a,b\mid a^6=b^4=1, a^3=b^2\rangle,
\\
\PSL_2(\mathbb{Z})&=\langle a,b\mid a^3=b^2=1\rangle.
\end{align*}
We also note that $S$ sends $z\in \mathbb{H}$ to  $-1/z$ and rotates $\mathcal{F}$ about the point $i$, fixing $i$ as shown in Figure~1. This is the only fixed point for the action of $S$ on $\mathbb{H}$. 
Furthermore, $S$ inverts the arc from $e^{2\pi i/3}$ to $e^{\pi i/3}$, while $TS$ fixes $e^{\pi i/3}$.

\begin{figure}\label{fig:SL2}
\begin{center}
\resizebox{5.89in}{4.88in}{\includegraphics[]{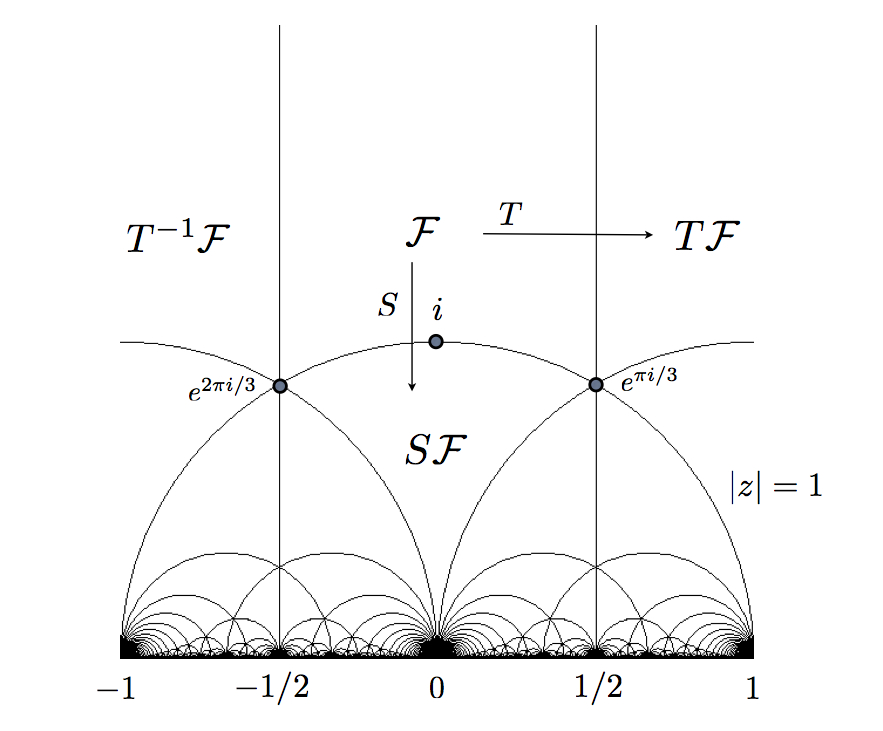}}
\caption{The fundamental domain for $\SL_2(\mathbb{Z})$ on $\mathbb{H}$}
\end{center}
\end{figure}

Since we will use the Steinberg presentation to generalize our results in the next section, we briefly describe the construction of $\SL_2(\R)$ as a Chevalley group.
Recall that the matrices 
$$e=\left(\begin{matrix} 0 & 1\\ 0 & 0\end{matrix}\right), \qquad h= \left(\begin{matrix} 1 & 0\\ 0& -1\end{matrix}\right),
\qquad f= \left(\begin{matrix} 0 & 0\\ 1 & 0\end{matrix}\right)$$
form a basis for the Lie algebra $\mathfrak{sl}_2(\mathbb{R})$ of the Lie group $\SL_2(\R)$.  The multiplication table for this Lie algebra is completely determined by the relations
$$[e,f]=h, \qquad  [h,e]=2e,\qquad  [h,f]=-2f.$$
Let  $\mathfrak{h}=\mathbb{R}h$, let
$\mathfrak{h}^*$ be the dual vector space, and let $\alpha\in\mathfrak{h}^*$ be defined by $\alpha(h)=2$.
Then $\mathfrak{sl}_2(\mathbb{R})$ has a {\it root space decomposition}, that is the eigenspace decomposition of $\mathfrak{sl}_2(\mathbb{R})$ with respect to 
the adjoint ad $h$, where  (ad $h)(X)=[h,X]$. 
The roots are $\alpha, -\alpha\in\mathfrak{h}^{\ast}$ with corresponding root vectors $e_{\alpha}=e$ and $e_{-\alpha}=f$.

The map
 $$\exp:\mathfrak{sl}_2(\mathbb{R})\longrightarrow \SL_2(\mathbb{R})$$
is the matrix exponential $\exp(X) = \sum_{k=0}^\infty{1 \over k!}X^k$. 
Now $\SL_2(\R)$ has Steinberg generators
\begin{align*}
\chi_{\alpha}(x)=\exp(xe_{\alpha})={\rm Id}+xe=\left(\begin{matrix} 1 & x\\ 0 & 1\end{matrix}\right)\quad\text{and}\quad
\chi_{-\alpha}(x)=\exp(xe_{-\alpha})={\rm Id}+xf=\left(\begin{matrix} 1 & 0\\ x & 1\end{matrix}\right),
\end{align*}
for $x\in\R$. We also define supplementary generators for $t\in\R^*$
\begin{align*}
\widetilde{w}_{\alpha}(t)&=\chi_{\alpha}(t)\chi_{-\alpha}(-t^{-1})\chi_{\alpha}(t)
=\left(\begin{matrix} 0 & t\\ -t^{-1} & 0\end{matrix}\right),\\
\widetilde{w}_{\alpha}&=\widetilde{w}_{\alpha}(1)=\left(\begin{matrix} 0 & 1\\ {-1} & 0\end{matrix}\right),\quad\text{and}\\
h_{\alpha}(t)&=\widetilde{w}_{\alpha}^{-1}\widetilde{w}_{\alpha}(t)=\left(\begin{matrix} t^{-1} & 0\\ 0& t \end{matrix}\right).
\end{align*}


Define the \emph{toral subgroup}
$$H=\{ h_{\alpha}(u)\mid u\in \R^{\ast}\},$$ which is isomorphic to $\R^*$.
Let $X_{\pm\alpha}=\{\chi_{\pm\alpha}(u)\mid u\in \R\}.$
Then $X_{\pm\alpha}$ is a subgroup isomorphic to the additive group of $\R$, called the {\it root group} associated with $\pm\alpha$.
We note that $\widetilde{w}_{\alpha}$ has order 4  and we set $\widetilde{W}=\langle \widetilde{w}_{\alpha}\mid \widetilde{w}_{\alpha}^4=1\rangle$. Then $\widetilde{W}$ is called the {\it extended Weyl group}. The  \emph{Weyl group} 
$$W=\langle w_{\alpha}\mid w_{\alpha}^2=1\rangle,$$ 
which is a symmetry group of the root system, is the quotient $\widetilde{W}/\widetilde{W}\cap H$.  

\begin{lemma}\label{relns}
$\SL_2(\mathbb{R})$ is generated by $\chi_\alpha(u)$, $\chi_{-\alpha}(u)$ for $u\in\R$; with supplementary generators $\widetilde{w}_\alpha(t)$, $\wt_\alpha$ and $h_{\alpha}(t)$ defined as above for $t\in\R^*$.
The Steinberg relations give defining relation for $\SL_2(\mathbb{R})$ with respect to these generators:
\begin{align}
h_{\alpha}(uv)&=h_{\alpha}(u)h_{\alpha}(v),\\
  \widetilde{w}_\alpha h_\alpha(t)\widetilde{w}_\alpha^{-1} &=  h_{\alpha}(t^{-1}),\\
\chi_{\alpha}(u+v)&=\chi_{\alpha}(u)\chi_{\alpha}(v),\\
\widetilde{w}_{\alpha}(t)\chi_{\alpha}(u)\widetilde{w}_{\alpha}(t)^{-1}&=\chi_{-\alpha}(-t^2u).
\end{align}
\end{lemma}






Expanding the last relation in terms of $\chi$'s, we get
\begin{equation}\label{sl2rel}
\chi_{-\alpha}(-t^{-1})\chi_{\alpha}(u)\chi_{-\alpha}(t^{-1})= \chi_{\alpha}(-t) \chi_{-\alpha}(-t^2u)\chi_{\alpha}(t).
\end{equation}
We note that this relation is required for $\SL_2(\mathbb{R})$, but for higher rank groups it 
is a consequence of the other Steinberg relations ([St2]).

We can now give an analogous the presentation for $\SL_2(\mathbb{Z})$. 
%
The main difference in the defining relations is that  $\mathbb{Z}^{\ast}=\{\pm 1\}$, which causes 
the  $\widetilde{w}_{\alpha}$  action on the torus (relation (2) above) to be trivial.  
\begin{proposition} 
 $\SL_2(\mathbb{Z})$ is generated by $h_\alpha(-1)$ and $\chi_{\pm\alpha}(n)$ for $n\in\Z$.
 Defining relations are
\begin{align}
\chi_{\pm\alpha}(n)&=\chi_{\pm\alpha}(1)^n,\\
%
\widetilde{w}_{\alpha}\chi_{\alpha}(\pm 1)\widetilde{w}_{\alpha}^{-1}&=\chi_{-\alpha}(\mp 1),\\
%
h_{\alpha}(-1)^2&=1.
\end{align}
\end{proposition}

It follows that $h_{\alpha}(-1)=\widetilde{w}_{\alpha}^2$. We note that for higher rank groups with more than one root, there are additional relations between the root groups, which will not always commute.

We now obtain Iwasawa coordinates for the coset space $\SO(2)\backslash \SL_2(\mathbb{R})$ in terms of these Chevalley generators. 
We use the Iwasawa decomposition: every $g\in \SL_2(\mathbb{R})$ has a unique representation as $g=khu$, where
\begin{align*}
k\in K&={\SO}(2),\qquad
h\in H^+=\left\{h_\alpha(t) \mid t\in \R, t>0\right\},\qquad
u\in U=X_{\alpha}.
\end{align*}
%
Now $Kg=Kkhu=Khu$, and so the coset $Kg\in K\backslash G$ is uniquely determined by $h=h_\alpha(t^{-1})\in H^+$ and $u=\chi_\alpha(x)\in U$. So we may represent an
element of $\SO(2)\backslash \SL_2(\mathbb{R})$ by its Iwasawa coordinates
\begin{equation}
hu
=\left(\begin{matrix} t & tx\\ 0 & t^{-1} \end{matrix}\right).
\label{eq zhu}
\end{equation}

We next consider the right action of $\SL_2(\mathbb{Z})$ on $\SO(2)\backslash \SL_2(\mathbb{R})$ 
in  Iwasawa coordinates. 
For $z \in \SL_2(\mathbb{Z})$, we compute $Kg\cdot z=Khu\cdot z$ and then try to use the relations from the Steinberg presentation to ensure that the result is again in  Iwasawa coordinates. It suffices to consider $z$ a generator of $\SL_2(\Z)$.
It is easy to see that the action of $S$ on $\SO(2)\backslash \SL_2(\mathbb{R})\cong\mathbb{H}$ has order 2. The element $S$ fixes the coset 
$\SO(2)\cdot \Id$ corresponding to the point $i\in\mathbb{H}$ and has no other fixed points in $\SO(2)\backslash \SL_2(\mathbb{R})$.  The element $T$ acts as a translation on $\mathbb{H}$ and hence has infinite order and acts without fixed points on $\SO(2)\backslash \SL_2(\mathbb{R})$. 
%
%
%
%
%

We note that $T=\chi_\alpha(1)
\in N$ and  $S=\widetilde{w}_\alpha
\in K$, since $\det(S)=1$ and $SS^T=S^TS={\rm Id}$.
The right action of $T$
is
$$Khu\cdot T
=K\left(\begin{matrix} t & t(x+1)\\ 0 & t^{-1} \end{matrix}\right)
=Kh_\alpha(t^{-1})\chi_\alpha(x+1),$$
which is in Iwasawa form. The right action of $S$
is
$$Khu\cdot S
=K\left(\begin{matrix} -tx & t\\ -t^{-1} & 0 \end{matrix}\right),$$
which is {\it not} in Iwasawa form. So we need a \emph{compensating element} $\kappa\in \SO(2)$ such that $\kappa hu\cdot S$ is upper triangular. 
\begin{proposition}\label{iwasawa}
A compensating element for the action of $S$ on $\SO(2)\backslash\SL_2(\R)$ is
$$\kappa=\dfrac{1}{\sqrt{t^2x^2+t^{-2}}}\left(\begin{matrix} -tx & -t^{-1}\\ t^{-1} & -tx \end{matrix}\right).$$
In Iwasawa coordinates the action is
$$Kh_\alpha(t^{-1})\chi_\alpha(x)\cdot S =Kh_\alpha\left(1/\sqrt{t^2x^2+t^{-2}}\right)\chi_\alpha\left( \dfrac{-t^2x }{{t^2x^2+t^{-2}}} \right).$$
%
\end{proposition}
\begin{proof}
Note that $\det(\kappa)=1$ and $\kappa\kappa^T=\kappa^T\kappa={\rm Id}$, thus $\kappa\in \SO(2)$. We have
 $$\kappa hu\cdot S=\dfrac{1}{\sqrt{t^2x^2+t^{-2}}}\left(\begin{matrix} -tx & -t^{-1}\\ t^{-1} & -tx \end{matrix}\right)\left(\begin{matrix} -tx & t\\ -t^{-1} & 0 \end{matrix}\right)
 =\dfrac{1}{\sqrt{t^2x^2+t^{-2}}}\left(\begin{matrix} t^2x^2+t^{-2} & -t^2x\\ 0 & 1 \end{matrix}\right)\;,$$
 which is in upper triangular form. 
 \end{proof}

An example of the right action of $S$ in Iwasawa coordinates is
$$\left(\begin{matrix} 1 & 0\\ 0 & 1\end{matrix}\right)\left(\begin{matrix} 1 & 1\\ 0 & 1\end{matrix}\right)\mapsto \left(\begin{matrix} \sqrt{2} & 0\\ 0 & \tfrac{1}{\sqrt{2}}\end{matrix}\right)\left(\begin{matrix} 1 & -\tfrac{1}{2}\\ 0 & 1\end{matrix}\right).$$

The physical significance of the compensating element $\kappa$ arises from its role as
 the element that 
restores the Borel gauge in the dimensional reduction of supergravity on 
a 2-torus $T^2$.

We now  write $\kappa$ in terms of the Steinberg generators.
The computer algebra system Magma (\cite{BCFS}) was used to find this formula, which can easily be verified by direct matrix multiplication.
\begin{proposition}
$\kappa=\chi_{\alpha}\left(-t^2x\right)h_{\alpha}\left(-1/\sqrt{t^4x^2+1}\right)\widetilde{w}_{\alpha}\chi_{\alpha}\left(-t^2x\right)$.
\end{proposition}

We finally  remark that  the general linear group can be similarly analyzed.  
In this case, the Steinberg generators for $\GL_2(\mathbb{R})$ are (\cite{CMT})
$$\chi_{\alpha}(x)=\left(\begin{matrix} 1 & x\\ 0 & 1\end{matrix}\right)\quad\text { and }\quad \chi_{-\alpha}(x)=\left(\begin{matrix} 1 & 0\\ x & 1\end{matrix}\right),$$
for $x\in\mathbb{R}$, and also two types of toral generator
$$h_1(t)=\left(\begin{matrix} t & 0\\ 0 & 1\end{matrix}\right)\quad\text { and }\quad h_2(t)=\left(\begin{matrix} 1 & 0\\ 0 & t\end{matrix}\right),$$
for $t\in \mathbb{R}^*$.

 \section{Iwasawa coordinates in a finite-dimensional Lie group}
 \label{Sec Iwasawa}
Let $G=G(\mathbb{R})$ be a split, simply connected, semisimple Lie group of rank $n$.  
Let $K$ be the maximal compact subgroup of $G$. 
In this section we present our method, using the Steinberg 
presentation, for computing Iwasawa coordinates of $K\backslash G$.  
We also compute the action of the generators of the extended Weyl group on $K\backslash G$, since these are the most interesting case 
as we saw in the last section.  
When $K\,\backslash G$ is 
a scalar coset for a supergravity theory, we apply this method in later sections 
to find the action of the integral form $G(\Z)$ of $G$ on $K\,\backslash G$.

Let $\Phi$ be the root system of $G$ and let $\Delta$ be a system of simple roots.
The Chevalley generators are $\chi_\alpha(x)$ for $\alpha\in\Phi$ and $x\in\R$.
We also define supplementary generators
$$\widetilde{w}_\beta(t) = \chi_{\beta}(t)\chi_{-\beta}(-t^{-1})\chi_{\beta}(t),\quad \widetilde{w}_\beta=\widetilde{w}_\beta(1),\quad h_\beta(t)=\widetilde{w}_\beta^{-1}\widetilde{w}_\beta(t)$$
for $\beta\in\Delta$ a simple root and $t\in\mathbb{R}^\times$. 
The precise Steinberg relations we are using are from  \cite{CMT} based on the method of Carter (\cite{Carter}):
\newcommand{\tens}{\otimes}
\begin{align}
  h_\beta(t)h_\beta(u) &= h_\beta(tu), \label{E-ha1eq}\\
  h_\alpha(t)h_\beta(u) &= h_\beta(u)h_\alpha(t), \label{E-ha1eq}\\
  \widetilde{w}_\alpha^{-1} h_\beta(t)\widetilde{w}_\alpha &= \prod_{\gamma\in\Delta} h_{\gamma}(t^{a_\gamma}) \qquad\text{where $w_\alpha\beta=\sum_{\gamma\in\Delta} a_\gamma \gamma$},\\
  \chi_\alpha(a) \chi_{\alpha}(b)  &= \chi_\alpha(a+b), \label{E-xcommeqs} \\
  \chi_{\beta}(-b)\chi_\alpha(a){\chi_{\beta}(b)} &= \chi_\alpha(a) \prod_{i,j>0} \chi_{i\alpha+j\beta}(C_{ij\alpha\beta}a^ib^j), \label{E-xswap}\\
  \chi_{-\alpha}(-t)\chi_\alpha(a){\chi_{-\alpha}(t)} &= \chi_{-\alpha}(-t^2a)^{\chi_\alpha(t^{-1})},  \label{E-rank1}\\
  \wt_\alpha^{-1}\chi_\beta(a)\wt_\alpha &=\chi_{w_\alpha\beta}(\eta_{\alpha\beta}a).
\end{align}
The structure constants $C_{ij\alpha\beta}$ and $\eta_{\alpha\beta}$ are as defined in \cite{CMT}, with all extraspecial structure constants equal to $(+1)$.
Note that equation~(\ref{E-rank1}) does not quite agree with the relation (\ref{sl2rel}) given in the last section -- this is a consequence of the fact that groups act on the right in Magma.

Let $\Phi^+=\{\alpha_1,\dots,\alpha_N\}$ denote the positive roots. We list the elements of $\Phi^+$ an  order consistent with height, that is, $\text{ht}(\alpha_i)< \text{ht}(\alpha_j)$ implies $i<j$. In particular, the simple roots are $\alpha_1,\dots,\alpha_n$.
Write $\kappa_i$ for $\kappa_{\alpha_i}$, $h_i$ for $h_{\alpha_i}$, $\chi_r$ for $\chi_{\alpha_r}$, and $\wt_i$ for $\wt_{\alpha_i}$.
Define the root subgroups $X_r=X_{\alpha_r}=\{\chi_r(a)\mid a\in \mathbb{R}\}$.
We define the unipotent subgroup $U=\langle X_\alpha\mid\alpha\in\Phi^+\rangle$, toral subgroup $H=\langle h_\alpha(t)\mid\alpha\in\Delta,t\in\R^*\rangle$, and 
$H^+=\langle h_\alpha(t)\mid\alpha\in\Delta,t\in\R,t>0\rangle$.
Let $\theta$ be a Cartan involution on $G$ defined by $\theta(\chi_\alpha(x))=\chi_{-\alpha}(-1/x)$ for all roots $\alpha\in\Phi$.  Then the maximal compact subgroup $K$ of $G$ is the stabilizer of $\theta$.  Then we get an Iwasawa decomposition $G=KH^+U$ as in the previous section, and every coset in $K\backslash G$ has Iwasawa coordinates
$Khu$ for unique $h\in H^+$ and $u\in U$.
We now have
\begin{lemma} \label{ic}
The Iwasawa coordinates for a coset in $K\backslash G$ are $Khu$ where
$$h=\prod_{j=1}^n h_j(t_j)\quad\text{and} \quad 
u = \prod_{r=1}^N \chi_r(x_r),$$
for some 
$t_j, x_r\in\R,$ with $t_j>0$.
\end{lemma}

Given a root $\beta\in\Phi$, the reflection $w_\beta$ has a reduced expression as a product of simple reflections $w_{\beta_1}\cdots w_{\beta_\ell}$,
and we define $\wt_\beta:= \widetilde{w}_{\beta_1}\cdots\widetilde{w}_{\beta_\ell}$. Note that this agrees with the previous definition for $\beta$ simple, and that $\wt_\beta$ was shown to be independent of the particular reduced expression used in \cite{CMT}.
Since $G$ is simply connected, for each $\beta\in\Phi$, the group $\langle  \chi_\beta(t), h_\beta(t), \wt_\beta\mid t\in\mathbb{R}^\times\rangle$ is isomorphic to $\SL_2(\mathbb{R})$.
We can therefore use the results of Section~\ref{S-action-sl2} for this special case to 
 finding the element $\kappa_\beta(x)\in K$ that makes $\kappa_\beta(x)\chi_\beta(x)\wt_\beta$ an element of 
$H^+U$, for a given root $\beta \in \Phi$.
\begin{lemma} If $\kappa_\beta(x) := \chi_{\beta}(-x)h_\beta(-1/\sqrt{x^2+1})\wt_\beta \chi_{\beta}(-x)$, then $\kappa_\beta(x)$ is in the maximal compact subgroup $K$, and 
$\kappa_\beta(x)\chi_\beta(x)\wt_\beta$ is in $H^+U$.
\end{lemma}

Given Iwasawa coordinates $hu$ as in Lemma~\ref{ic} and $\alpha_i\in\Delta$, we now find $\kappa\in K$ such that $\kappa hu \wt_i\in H^+U$.
First we can rewrite $u= \chi_i(x_i)\tilde{u}$, where
$$\tilde{u} =  \prod_{r=1}^{i-1} \chi_r(x_r)^{\chi_i(x_i)}\cdot \prod_{r=i+1}^N \chi_r(x_r).$$
Applying the Steinberg commutator relations we get
$$\chi_r(x_r)^{\chi_i(x_i)}= \chi_r(x_r)\prod_{\substack{a,b>0\text{ s.t.}\\a\alpha_r+b\alpha_i\in\Phi}} \chi_{a\alpha_r+b\alpha_i}(C_{abri}x_r^ax_i^b),$$
and hence $\tilde{u}\in\prod_{r\ne i} X_r$.
We now rearrange
$$hu\wt_i = h \chi_i(x_i) \tilde{u} \wt_i = 
 \chi_i(x_i)^{h^{-1}}h\tilde{u} \wt_i =  \chi_i(x_i)^{h^{-1}} \wt_i (h\tilde{u})^{\wt_i}.$$
Notice that  $(h\tilde{u})^{\wt_i}\in H^+U$, since $h^{\wt_i}\in H^+$, and $\tilde{u}^{\wt_i}\in(\prod_{r\ne i}X_r)^{\wt_i} =\prod_{r\ne i}X_{w_i\alpha_r}\subseteq U$ because $w_i\alpha_r\in\Phi^+$ when $r\ne i$.
Also
$\chi_i(x_i)^{h^{-1}} = \chi_i\left(c_i x_i\right),$
where $c_i:=\prod_{j=1}^nt_j^{-C_{ij}}$ for $C=(C_{ij})$ the Cartan matrix.
We can now choose $\kappa=\kappa_i(c_ix_i)$ to ensure that $\kappa hu\wt_i=\kappa\chi_i\left(c_i x_i\right)\wt_i \cdot (h\tilde{u})^{\wt_i}$ is indeed in $H^+U$.

We now give an explicit formula for $\kappa hu\wt_i$.
\begin{proposition} Given $h=\prod_{j=1}^n h_j(t_j)\in H^+$ and $u = \prod_{r=1}^N \chi_r(x_r)\in U$, we have
$$\kappa hu \cdot \wt_i= h^{(i)}u^{(i)},$$
where
\begin{align*}
 h^{(i)} &=h_i(c_i/r_i)h\in H^+, \quad\text{and}\\
 u^{(i)}&=\chi_{i}\left(-\left({c_i}/{r_i}\right)^2x_i\right)
\cdot\prod_{r=1}^{i-1} \chi_{w_i\alpha_r}(\eta_{ri}x_i) \tilde{\chi}_{ri} \cdot \prod_{r=i+1}^N \chi_{w_i\alpha_r}(\eta_{ri}x_i)\in U,
\end{align*}
for $c_i=\prod_{j=1}^nt_j^{-C_{ij}}$, $r_i=\sqrt{(c_ix_i)^2+1}$, and
$$ \tilde{\chi}_{ri} =  
\prod_{\substack{a,b>0\text{ s.t.}\\a\alpha_r+b\alpha_i\in\Phi}}
\chi_{w_i(a\alpha_{r}+b\alpha_i)}(\eta_{a\alpha_{r}+b\alpha_i,\alpha_i}C_{abri}x_r^ax_i^b).$$
\end{proposition}

\begin{proof} We have
\begin{eqnarray*}
\kappa hu \wt_i= \kappa_{i}(c_ix_i)hu\wt_i 
&=& \chi_{i}(-c_ix_i)h_i(-1/r_i)\wt_i \chi_{i}(-c_ix_i)\cdot
\chi_{i}(c_ix_i) \wt_i \left(h\tilde{u}\right)^{\wt_i}\\
&=& \chi_{i}(-c_ix_i)h_i(-1/r_i)\wt_i^2 \left(h \tilde{u}\right)^{\wt_i}\\
&=& \chi_{i}(-c_ix_i)\left(h_i(-1/r_i) h_i(-1) h^{\wt_i}\right) \tilde{u}^{\wt_i}\\
&=& h^{(i)} u^{(i)}
\end{eqnarray*}
where
$$
h^{(i)} := h_i(-1/r_i) h_i(-1) h^{\wt_i} = h_i(-1/r_i) h_i(-1) h_i(c_i) h=
 h_i(c_i/r_i)h
$$
and
\begin{eqnarray*}
u^{(i)} &:=& \chi_{i}(-c_ix_i)^{h^{(i)}} \tilde{u}^{\wt_i}\\
&=& \chi_{i}\left( (c_i/r_i)^2 \prod_{j=1}^nt_j^{C_{ij}}\cdot(-c_ix_i)\right)
\left(\prod_{r=1}^{i-1} \chi_r(x_r)^{\chi_i(x_i)}\cdot \prod_{r=i+1}^N \chi_r(x_r) \right)^{\wt_i}\\
&=& \chi_{i}\left(-\left({c_i}/{r_i}\right)^2x_i\right)
\cdot\prod_{r=1}^{i-1} \chi_r(x_r)^{\chi_i(x_i)\wt_i}\cdot \prod_{r=i+1}^N \chi_r(x_r)^{\wt_i}\\
&=&\chi_{i}\left(-\left({c_i}/{r_i}\right)^2x_i\right)
\cdot\prod_{r=1}^{i-1} \chi_{w_i\alpha_r}(\eta_{ri}x_i) \tilde{\chi}_{ri} \cdot \prod_{r=i+1}^N \chi_{w_i\alpha_r}(\eta_{ri}x_i).
\end{eqnarray*}
\end{proof}

We can find the Iwasawa coordinates for the expression given in this proposition using the symbolic collection algorithm in \cite{CMH}. 
These coordinates are easy to store in Magma, but too large to list here.
Instead we give the uncollected expression for the Iwasawa coordinates for the action of generators of the extended Weyl group on supergravity cosets for $n=2,\dots,8$.  
Since each root system is contained in the previous one, we can fix a single root ordering for all of them.  The root indices for each value of $n$ are given in Table~\ref{T-RtIndices}. 
Tables~\ref{T-first} to \ref{T-last} give expressions for $\kappa_i(c_ix_i)hu\wt_i$ for $i=1,\dots,n$ in the supergravity coset space of rank $n$.

\begin{table}{\tiny
$\begin{array}{llllllll|llllllll|llllllll}
\multicolumn{7}{ l }{\text{Root indices for $n=$}}& &\multicolumn{7}{ l }{\text{Root indices for $n=$}}& &\multicolumn{7}{ l }{\text{Root indices for $n=$}}& \\
2&3&4&5&6&7&8& \text{Root}& 2&3&4&5&6&7&8& \text{Root}& 2&3&4&5&6&7&8& \text{Root} \\\hline
1 & 1 & 1 & 1 & 1 & 1 & 1 & (1 0 0 0 0 0 0 0) &  &  &  &  &  & 36 & 41 & (0 1 1 1 1 1 1 0) &  &  &  &  &  &  & 81 & (1 2 2 3 2 1 1 1)\\
2 & 2 & 2 & 2 & 2 & 2 & 2 & (0 1 0 0 0 0 0 0) &  &  &  &  &  &  & 42 & (0 0 1 1 1 1 1 1) &  &  &  &  &  & 59 & 82 & (1 2 1 3 3 2 1 0)\\
 & 3 & 3 & 3 & 3 & 3 & 3 & (0 0 1 0 0 0 0 0) &  &  &  &  &  &  & 43 & (0 1 0 1 1 1 1 1) &  &  &  &  &  &  & 83 & (1 2 1 3 2 2 1 1)\\
 &  & 4 & 4 & 4 & 4 & 4 & (0 0 0 1 0 0 0 0) &  &  &  & 20 & 29 & 37 & 44 & (1 2 1 2 1 0 0 0) &  &  &  &  &  &  & 84 & (1 2 1 2 2 2 2 1)\\
 &  &  & 5 & 5 & 5 & 5 & (0 0 0 0 1 0 0 0) &  &  &  &  & 30 & 38 & 45 & (1 1 1 2 1 1 0 0) &  &  &  &  &  & 60 & 85 & (1 2 2 3 3 2 1 0)\\
 &  &  &  & 6 & 6 & 6 & (0 0 0 0 0 1 0 0) &  &  &  &  &  & 39 & 46 & (1 1 1 1 1 1 1 0) &  &  &  &  &  &  & 86 & (1 2 2 3 2 2 1 1)\\
 &  &  &  &  & 7 & 7 & (0 0 0 0 0 0 1 0) &  &  &  &  &  &  & 47 & (1 1 0 1 1 1 1 1) &  &  &  &  &  &  & 87 & (1 2 1 3 3 2 1 1)\\
 &  &  &  &  &  & 8 & (0 0 0 0 0 0 0 1) &  &  &  &  & 31 & 40 & 48 & (0 1 1 2 2 1 0 0) &  &  &  &  &  &  & 88 & (1 2 1 3 2 2 2 1)\\
3 & 4 & 5 & 6 & 7 & 8 & 9 & (1 1 0 0 0 0 0 0) &  &  &  &  &  & 41 & 49 & (0 1 1 2 1 1 1 0) &  &  &  &  &  & 61 & 89 & (1 2 2 4 3 2 1 0)\\
 &  & 7 & 8 & 8 & 9 & 10 & (0 0 1 1 0 0 0 0) &  &  &  &  &  &  & 50 & (0 1 1 1 1 1 1 1) &  &  &  &  &  &  & 90 & (1 2 2 3 3 2 1 1)\\
 &  & 6 & 7 & 9 & 10 & 11 & (0 1 0 1 0 0 0 0) &  &  &  &  & 32 & 42 & 51 & (1 2 1 2 1 1 0 0) &  &  &  &  &  &  & 91 & (1 2 2 3 2 2 2 1)\\
 &  &  & 9 & 10 & 11 & 12 & (0 0 0 1 1 0 0 0) &  &  &  &  & 33 & 43 & 52 & (1 1 1 2 2 1 0 0) &  &  &  &  &  &  & 92 & (1 2 1 3 3 2 2 1)\\
 &  &  &  & 11 & 12 & 13 & (0 0 0 0 1 1 0 0) &  &  &  &  &  & 44 & 53 & (1 1 1 2 1 1 1 0) &  &  &  &  &  & 62 & 93 & (1 3 2 4 3 2 1 0)\\
 &  &  &  &  & 13 & 14 & (0 0 0 0 0 1 1 0) &  &  &  &  &  &  & 54 & (1 1 1 1 1 1 1 1) &  &  &  &  &  &  & 94 & (1 2 2 4 3 2 1 1)\\
 &  &  &  &  &  & 15 & (0 0 0 0 0 0 1 1) &  &  &  &  &  & 45 & 55 & (0 1 1 2 2 1 1 0) &  &  &  &  &  &  & 95 & (1 2 2 3 3 2 2 1)\\
 &  & 8 & 10 & 12 & 14 & 16 & (1 1 0 1 0 0 0 0) &  &  &  &  &  &  & 56 & (0 1 1 2 1 1 1 1) &  &  &  &  &  &  & 96 & (1 2 1 3 3 3 2 1)\\
 &  & 9 & 11 & 13 & 15 & 17 & (0 1 1 1 0 0 0 0) &  &  &  &  & 34 & 46 & 57 & (1 2 1 2 2 1 0 0) &  &  &  &  &  & 63 & 97 & (2 3 2 4 3 2 1 0)\\
 &  &  & 13 & 14 & 16 & 18 & (0 0 1 1 1 0 0 0) &  &  &  &  &  & 47 & 58 & (1 2 1 2 1 1 1 0) &  &  &  &  &  &  & 98 & (1 3 2 4 3 2 1 1)\\
 &  &  & 12 & 15 & 17 & 19 & (0 1 0 1 1 0 0 0) &  &  &  &  &  & 48 & 59 & (1 1 1 2 2 1 1 0) &  &  &  &  &  &  & 99 & (1 2 2 4 3 2 2 1)\\
 &  &  &  & 16 & 18 & 20 & (0 0 0 1 1 1 0 0) &  &  &  &  &  &  & 60 & (1 1 1 2 1 1 1 1) &  &  &  &  &  &  & 100 & (1 2 2 3 3 3 2 1)\\
 &  &  &  &  & 19 & 21 & (0 0 0 0 1 1 1 0) &  &  &  &  &  & 49 & 61 & (0 1 1 2 2 2 1 0) &  &  &  &  &  &  & 101 & (2 3 2 4 3 2 1 1)\\
 &  &  &  &  &  & 22 & (0 0 0 0 0 1 1 1) &  &  &  &  &  &  & 62 & (0 1 1 2 2 1 1 1) &  &  &  &  &  &  & 102 & (1 3 2 4 3 2 2 1)\\
 &  & 10 & 14 & 17 & 20 & 23 & (1 1 1 1 0 0 0 0) &  &  &  &  & 35 & 50 & 63 & (1 2 1 3 2 1 0 0) &  &  &  &  &  &  & 103 & (1 2 2 4 3 3 2 1)\\
 &  &  & 15 & 18 & 21 & 24 & (1 1 0 1 1 0 0 0) &  &  &  &  &  & 51 & 64 & (1 2 1 2 2 1 1 0) &  &  &  &  &  &  & 104 & (2 3 2 4 3 2 2 1)\\
 &  &  & 16 & 19 & 22 & 25 & (0 1 1 1 1 0 0 0) &  &  &  &  &  &  & 65 & (1 2 1 2 1 1 1 1) &  &  &  &  &  &  & 105 & (1 3 2 4 3 3 2 1)\\
 &  &  &  & 20 & 23 & 26 & (0 0 1 1 1 1 0 0) &  &  &  &  &  & 52 & 66 & (1 1 1 2 2 2 1 0) &  &  &  &  &  &  & 106 & (1 2 2 4 4 3 2 1)\\
 &  &  &  & 21 & 24 & 27 & (0 1 0 1 1 1 0 0) &  &  &  &  &  &  & 67 & (1 1 1 2 2 1 1 1) &  &  &  &  &  &  & 107 & (2 3 2 4 3 3 2 1)\\
 &  &  &  &  & 25 & 28 & (0 0 0 1 1 1 1 0) &  &  &  &  &  &  & 68 & (0 1 1 2 2 2 1 1) &  &  &  &  &  &  & 108 & (1 3 2 4 4 3 2 1)\\
 &  &  &  &  &  & 29 & (0 0 0 0 1 1 1 1) &  &  &  &  & 36 & 53 & 69 & (1 2 2 3 2 1 0 0) &  &  &  &  &  &  & 109 & (2 3 2 4 4 3 2 1)\\
 &  &  & 17 & 22 & 26 & 30 & (1 1 1 1 1 0 0 0) &  &  &  &  &  & 54 & 70 & (1 2 1 3 2 1 1 0) &  &  &  &  &  &  & 110 & (1 3 2 5 4 3 2 1)\\
 &  &  &  & 23 & 27 & 31 & (1 1 0 1 1 1 0 0) &  &  &  &  &  & 55 & 71 & (1 2 1 2 2 2 1 0) &  &  &  &  &  &  & 111 & (2 3 2 5 4 3 2 1)\\
 &  &  & 18 & 24 & 28 & 32 & (0 1 1 2 1 0 0 0) &  &  &  &  &  &  & 72 & (1 2 1 2 2 1 1 1) &  &  &  &  &  &  & 112 & (1 3 3 5 4 3 2 1)\\
 &  &  &  & 25 & 29 & 33 & (0 1 1 1 1 1 0 0) &  &  &  &  &  &  & 73 & (1 1 1 2 2 2 1 1) &  &  &  &  &  &  & 113 & (2 3 3 5 4 3 2 1)\\
 &  &  &  &  & 30 & 34 & (0 0 1 1 1 1 1 0) &  &  &  &  &  &  & 74 & (0 1 1 2 2 2 2 1) &  &  &  &  &  &  & 114 & (2 4 2 5 4 3 2 1)\\
 &  &  &  &  & 31 & 35 & (0 1 0 1 1 1 1 0) &  &  &  &  &  & 56 & 75 & (1 2 2 3 2 1 1 0) &  &  &  &  &  &  & 115 & (2 4 3 5 4 3 2 1)\\
 &  &  &  &  &  & 36 & (0 0 0 1 1 1 1 1) &  &  &  &  &  & 57 & 76 & (1 2 1 3 2 2 1 0) &  &  &  &  &  &  & 116 & (2 4 3 6 4 3 2 1)\\
 &  &  & 19 & 26 & 32 & 37 & (1 1 1 2 1 0 0 0) &  &  &  &  &  &  & 77 & (1 2 1 3 2 1 1 1) &  &  &  &  &  &  & 117 & (2 4 3 6 5 3 2 1)\\
 &  &  &  & 27 & 33 & 38 & (1 1 1 1 1 1 0 0) &  &  &  &  &  &  & 78 & (1 2 1 2 2 2 1 1) &  &  &  &  &  &  & 118 & (2 4 3 6 5 4 2 1)\\
 &  &  &  &  & 34 & 39 & (1 1 0 1 1 1 1 0) &  &  &  &  &  &  & 79 & (1 1 1 2 2 2 2 1) &  &  &  &  &  &  & 119 & (2 4 3 6 5 4 3 1)\\
 &  &  &  & 28 & 35 & 40 & (0 1 1 2 1 1 0 0) &  &  &  &  &  & 58 & 80 & (1 2 2 3 2 2 1 0) &  &  &  &  &  &  & 120 & (2 4 3 6 5 4 3 2)
\end{array}$}
\medskip\caption{Root indices}\label{T-RtIndices}
\end{table}

\begin{table} {\tiny
$\begin{array}{lll} i & c_i & \kappa_i{}(c_ix_i)hu\wt_i \\\hline
1 & t_{2}/t_{1}^2 & h_{1}(c_1/r_1)h\cdot \chi_{1}(-(c_1/r_1)^2x_1)\cdot\chi_{3}(-x_{2})  \chi_{2}(x_{3})   \\
2 & t_{1}/t_{2}^2 & h_{2}(c_2/r_2)h\cdot \chi_{2}(-(c_2/r_2)^2x_2)\cdot\chi_{1}(-x_{1}x_{2}) \chi_{3}(x_{1})  \cdot \chi_{1}(-x_{3})   \\
\end{array}$}
\medskip\caption{Iwasawa coordinates for $A_2$ with compensating element 
\mbox{$\kappa_i(c_ix_i)=\chi_i(-c_ix_i)h_i(-1/\sqrt{(c_ix_i)^2+1})\wt_i \chi_{i}(-c_ix_i)$}}\label{T-first}
\end{table}

\begin{table} {\tiny
$\begin{array}{lll} i & c_i & \kappa_i{}(c_ix_i)hu\wt_i \\\hline

1 & t_{2}/t_{1}^2 & h_{1}(c_1/r_1)h\cdot \chi_{i}(-{}(c_1/r_1)^2x_1)\cdot\chi_{4}(-x_{2})  \chi_{3}(x_{3})  \chi_{2}(x_{4})   \\
2 & t_{1}/t_{2}^2 & h_{2}(c_2/r_2)h\cdot \chi_{i}(-{}(c_2/r_2)^2x_2)\cdot\chi_{1}(-x_{1}x_{2}) \chi_{4}(x_{1})  \cdot \chi_{3}(x_{3})  \chi_{1}(-x_{4})   \\
3 & 1/t_{3}^2 & h_{3}(c_3/r_3)h\cdot \chi_{i}(-{}(c_3/r_3)^2x_3)\cdot\chi_{1}(x_{1})  \chi_{2}(x_{2})  \cdot \chi_{4}(x_{4})   \\
\end{array}$
}\medskip\caption{Iwasawa coordinates for $A_2A_1$ with compensating element 
\mbox{$\kappa_i(c_ix_i)=\chi_i(-c_ix_i)h_i(-1/\sqrt{(c_ix_i)^2+1})\wt_i \chi_{i}(-c_ix_i)$}}
\end{table}

\begin{table} {\tiny
$\begin{array}{lll} i & c_i & \kappa_i{}(c_ix_i)hu\wt_i \\\hline

1 & t_{2}/t_{1}^2 & h_{1}(c_1/r_1)h\cdot \chi_{i}(-{}(c_1/r_1)^2x_1)\cdot\chi_{5}(-x_{2})  \chi_{3}(x_{3})  \chi_{4}(x_{4})  \chi_{2}(x_{5})  \chi_{8}(-x_{6})  \chi_{7}(x_{7})  \chi_{6}(x_{8})  \chi_{10}(-x_{9})  \chi_{9}(x_{10})   \\
2 & t_{1}t_{3}/t_{2}^2 & h_{2}(c_2/r_2)h\cdot \chi_{i}(-{}(c_2/r_2)^2x_2)\cdot\chi_{1}(-x_{1}x_{2}) \chi_{5}(x_{1})  \cdot \chi_{6}(-x_{3})  \chi_{4}(x_{4})  \chi_{1}(-x_{5})  \chi_{3}(x_{6})  \chi_{9}(-x_{7})  \chi_{8}(x_{8})  \chi_{7}(x_{9})  \chi_{10}(x_{10})   \\
3 & t_{2}t_{4}/t_{3}^2 & h_{3}(c_3/r_3)h\cdot \chi_{i}(-{}(c_3/r_3)^2x_3)\cdot\chi_{1}(x_{1})  \chi_{2}(-x_{2}x_{3}) \chi_{6}(x_{2})  \cdot \chi_{7}(-x_{4})  \chi_{8}(x_{5})  \chi_{2}(-x_{6})  \chi_{4}(x_{7})  \chi_{5}(-x_{8})  \chi_{9}(x_{9})  \chi_{10}(x_{10})   \\
4 & t_{3}/t_{4}^2 & h_{4}(c_4/r_4)h\cdot \chi_{i}(-{}(c_4/r_4)^2x_4)\cdot\chi_{1}(x_{1})  \chi_{2}(x_{2})  \chi_{3}(-x_{3}x_{4}) \chi_{7}(x_{3})  \cdot \chi_{5}(x_{5})  \chi_{9}(x_{6})  \chi_{3}(-x_{7})  \chi_{10}(x_{8})  \chi_{6}(-x_{9})  \chi_{8}(-x_{10})   \\
\end{array}$
}\medskip\caption{Iwasawa coordinates for $A_4$ with compensating element 
\mbox{$\kappa_i(c_ix_i)=\chi_i(-c_ix_i)h_i(-1/\sqrt{(c_ix_i)^2+1})\wt_i \chi_{i}(-c_ix_i)$}}
\end{table}
\begin{table} {\tiny
$\begin{array}{lll} i & c_i & \kappa_i{}(c_ix_i)hu\wt_i \\\hline

1 & t_{2}x_{1}/t_{1}^2 & h_{1}(c_1/r_1)h\cdot \chi_{i}(-{}(c_1/r_1)^2x_1)\cdot\chi_{6}(-x_{2})  \chi_{3}(x_{3})  \chi_{4}(x_{4})  \chi_{5}(x_{5})  \chi_{2}(x_{6})  \chi_{10}(-x_{7})  \chi_{8}(x_{8})  \chi_{9}(x_{9})  \chi_{7}(x_{10})  \\
&&\chi_{14}(-x_{11})  \chi_{15}(-x_{12})  \chi_{13}(x_{13})  \chi_{11}(x_{14})  \chi_{12}(x_{15})  \chi_{17}(-x_{16})  \chi_{16}(x_{17})  \chi_{19}(-x_{18})  \chi_{18}(x_{19})  \chi_{20}(x_{20})   \\
2 & t_{1}t_{3}x_{2}/t_{2}^2 & h_{2}(c_2/r_2)h\cdot \chi_{i}(-{}(c_2/r_2)^2x_2)\cdot\chi_{1}(-x_{1}x_{2}) \chi_{6}(x_{1})  \cdot \chi_{7}(-x_{3})  \chi_{4}(x_{4})  \chi_{5}(x_{5})  \chi_{1}(-x_{6})  \chi_{3}(x_{7})  \chi_{11}(-x_{8})  \chi_{12}(-x_{9})  \chi_{10}(x_{10})  \\
&&\chi_{8}(x_{11})  \chi_{9}(x_{12})  \chi_{16}(-x_{13})  \chi_{14}(x_{14})  \chi_{15}(x_{15})  \chi_{13}(x_{16})  \chi_{17}(x_{17})  \chi_{18}(x_{18})  \chi_{20}(-x_{19})  \chi_{19}(x_{20})   \\
3 & t_{2}t_{4}t_{5}x_{3}/t_{3}^2 & h_{3}(c_3/r_3)h\cdot \chi_{i}(-{}(c_3/r_3)^2x_3)\cdot\chi_{1}(x_{1})  \chi_{2}(-x_{2}x_{3}) \chi_{7}(x_{2})  \cdot \chi_{8}(-x_{4})  \chi_{9}(-x_{5})  \chi_{10}(x_{6})  \chi_{2}(-x_{7})  \chi_{4}(x_{8})  \chi_{5}(x_{9})  \chi_{6}(-x_{10})  \\
&&\chi_{11}(x_{11})  \chi_{12}(x_{12})  \chi_{13}(x_{13})  \chi_{14}(x_{14})  \chi_{15}(x_{15})  \chi_{18}(-x_{16})  \chi_{19}(-x_{17})  \chi_{16}(x_{18})  \chi_{17}(x_{19})  \chi_{20}(x_{20})   \\
4 & t_{3}x_{4}/t_{4}^2 & h_{4}(c_4/r_4)h\cdot \chi_{i}(-{}(c_4/r_4)^2x_4)\cdot\chi_{1}(x_{1})  \chi_{2}(x_{2})  \chi_{3}(-x_{3}x_{4}) \chi_{8}(x_{3})  \cdot \chi_{5}(x_{5})  \chi_{6}(x_{6})  \chi_{11}(x_{7})  \chi_{3}(-x_{8})  \chi_{13}(-x_{9})  \chi_{14}(x_{10})  \\
&&\chi_{7}(-x_{11})  \chi_{16}(-x_{12})  \chi_{9}(x_{13})  \chi_{10}(-x_{14})  \chi_{17}(-x_{15})  \chi_{12}(x_{16})  \chi_{15}(x_{17})  \chi_{18}(x_{18})  \chi_{19}(x_{19})  \chi_{20}(x_{20})   \\
5 & t_{3}x_{5}/t_{5}^2 & h_{5}(c_5/r_5)h\cdot \chi_{i}(-{}(c_5/r_5)^2x_5)\cdot\chi_{1}(x_{1})  \chi_{2}(x_{2})  \chi_{3}(-x_{3}x_{5}) \chi_{9}(x_{3})  \chi_{4}(x_{4})  \cdot \chi_{6}(x_{6})  \chi_{12}(x_{7})  \chi_{13}(-x_{8})  \chi_{3}(-x_{9})  \chi_{15}(x_{10})  \\
&&\chi_{16}(-x_{11})  \chi_{7}(-x_{12})  \chi_{8}(x_{13})  \chi_{17}(-x_{14})  \chi_{10}(-x_{15})  \chi_{11}(x_{16})  \chi_{14}(x_{17})  \chi_{18}(x_{18})  \chi_{19}(x_{19})  \chi_{20}(x_{20})   \\
\end{array}$
}\medskip\caption{Iwasawa coordinates for $D_5$ with compensating element 
\mbox{$\kappa_i(c_ix_i)=\chi_i(-c_ix_i)h_\beta(-1/\sqrt{(c_ix_i)^2+1})\wt_i \chi_{i}(-c_ix_i)$}}
\end{table}
\begin{table} {\tiny
$\begin{array}{lll} i & c_i & \kappa_i{}(c_ix_i)hu\wt_i \\\hline

1 & t_{3}x_{1}/t_{1}^2 & h_{1}(c_1/r_1)h\cdot \chi_{i}(-{}(c_1/r_1)^2x_1)\cdot\chi_{2}(x_{2})  \chi_{7}(-x_{3})  \chi_{4}(x_{4})  \chi_{5}(x_{5})  \chi_{6}(x_{6})  \chi_{3}(x_{7})  \chi_{8}(x_{8})  \chi_{12}(-x_{9})  \chi_{10}(x_{10})  \\
&&\chi_{11}(x_{11})  \chi_{9}(x_{12})  \chi_{17}(-x_{13})  \chi_{14}(x_{14})  \chi_{18}(-x_{15})  \chi_{16}(x_{16})  \chi_{13}(x_{17})  \chi_{15}(x_{18})  \chi_{22}(-x_{19})  \chi_{20}(x_{20})  \\
&&\chi_{23}(-x_{21})  \chi_{19}(x_{22})  \chi_{21}(x_{23})  \chi_{26}(-x_{24})  \chi_{27}(-x_{25})  \chi_{24}(x_{26})  \chi_{25}(x_{27})  \chi_{30}(-x_{28})  \chi_{29}(x_{29})  \chi_{28}(x_{30})  \\
&&\chi_{33}(-x_{31})  \chi_{32}(x_{32})  \chi_{31}(x_{33})  \chi_{34}(x_{34})  \chi_{35}(x_{35})  \chi_{36}(x_{36})   \\
2 & t_{4}x_{2}/t_{2}^2 & h_{2}(c_2/r_2)h\cdot \chi_{i}(-{}(c_2/r_2)^2x_2)\cdot\chi_{1}(x_{1})  \cdot \chi_{3}(x_{3})  \chi_{8}(-x_{4})  \chi_{5}(x_{5})  \chi_{6}(x_{6})  \chi_{7}(x_{7})  \chi_{4}(x_{8})  \chi_{13}(-x_{9})  \chi_{14}(-x_{10})  \\
&&\chi_{11}(x_{11})  \chi_{17}(-x_{12})  \chi_{9}(x_{13})  \chi_{10}(x_{14})  \chi_{19}(-x_{15})  \chi_{20}(-x_{16})  \chi_{12}(x_{17})  \chi_{22}(-x_{18})  \chi_{15}(x_{19})  \chi_{16}(x_{20})  \\
&&\chi_{25}(-x_{21})  \chi_{18}(x_{22})  \chi_{27}(-x_{23})  \chi_{24}(x_{24})  \chi_{21}(x_{25})  \chi_{26}(x_{26})  \chi_{23}(x_{27})  \chi_{28}(x_{28})  \chi_{29}(x_{29})  \chi_{30}(x_{30})  \\
&&\chi_{31}(x_{31})  \chi_{32}(x_{32})  \chi_{33}(x_{33})  \chi_{34}(x_{34})  \chi_{36}(-x_{35})  \chi_{35}(x_{36})   \\
3 & t_{1}t_{4}x_{3}/t_{3}^2 & h_{3}(c_3/r_3)h\cdot \chi_{i}(-{}(c_3/r_3)^2x_3)\cdot\chi_{1}(-x_{1}x_{3}) \chi_{7}(x_{1})  \chi_{2}(x_{2})  \cdot \chi_{9}(-x_{4})  \chi_{5}(x_{5})  \chi_{6}(x_{6})  \chi_{1}(-x_{7})  \chi_{13}(-x_{8})  \chi_{4}(x_{9})  \chi_{15}(-x_{10})  \\
&&\chi_{11}(x_{11})  \chi_{12}(x_{12})  \chi_{8}(x_{13})  \chi_{19}(-x_{14})  \chi_{10}(x_{15})  \chi_{21}(-x_{16})  \chi_{17}(x_{17})  \chi_{18}(x_{18})  \chi_{14}(x_{19})  \chi_{25}(-x_{20})  \\
&&\chi_{16}(x_{21})  \chi_{22}(x_{22})  \chi_{23}(x_{23})  \chi_{24}(x_{24})  \chi_{20}(x_{25})  \chi_{29}(-x_{26})  \chi_{27}(x_{27})  \chi_{28}(x_{28})  \chi_{26}(x_{29})  \chi_{32}(-x_{30})  \\
&&\chi_{31}(x_{31})  \chi_{30}(x_{32})  \chi_{34}(-x_{33})  \chi_{33}(x_{34})  \chi_{35}(x_{35})  \chi_{36}(x_{36})   \\
4 & t_{2}t_{3}t_{5}x_{4}/t_{4}^2 & h_{4}(c_4/r_4)h\cdot \chi_{i}(-{}(c_4/r_4)^2x_4)\cdot\chi_{1}(x_{1})  \chi_{2}(-x_{2}x_{4}) \chi_{8}(x_{2})  \chi_{3}(-x_{3}x_{4}) \chi_{9}(x_{3})  \cdot \chi_{10}(-x_{5})  \chi_{6}(x_{6})  \chi_{12}(x_{7})  \chi_{2}(-x_{8})  \chi_{3}(-x_{9})  \chi_{5}(x_{10})  \\
&&\chi_{16}(-x_{11})  \chi_{7}(-x_{12})  \chi_{13}(x_{13})  \chi_{14}(x_{14})  \chi_{15}(x_{15})  \chi_{11}(x_{16})  \chi_{17}(x_{17})  \chi_{18}(x_{18})  \chi_{24}(-x_{19})  \chi_{20}(x_{20})  \\
&&\chi_{21}(x_{21})  \chi_{26}(-x_{22})  \chi_{23}(x_{23})  \chi_{19}(x_{24})  \chi_{28}(-x_{25})  \chi_{22}(x_{26})  \chi_{30}(-x_{27})  \chi_{25}(x_{28})  \chi_{29}(x_{29})  \chi_{27}(x_{30})  \\
&&\chi_{31}(x_{31})  \chi_{32}(x_{32})  \chi_{33}(x_{33})  \chi_{35}(-x_{34})  \chi_{34}(x_{35})  \chi_{36}(x_{36})   \\
5 & t_{4}t_{6}x_{5}/t_{5}^2 & h_{5}(c_5/r_5)h\cdot \chi_{i}(-{}(c_5/r_5)^2x_5)\cdot\chi_{1}(x_{1})  \chi_{2}(x_{2})  \chi_{3}(x_{3})  \chi_{4}(-x_{4}x_{5}) \chi_{10}(x_{4})  \cdot \chi_{11}(-x_{6})  \chi_{7}(x_{7})  \chi_{14}(x_{8})  \chi_{15}(x_{9})  \chi_{4}(-x_{10})  \\
&&\chi_{6}(x_{11})  \chi_{18}(x_{12})  \chi_{19}(x_{13})  \chi_{8}(-x_{14})  \chi_{9}(-x_{15})  \chi_{16}(x_{16})  \chi_{22}(x_{17})  \chi_{12}(-x_{18})  \chi_{13}(-x_{19})  \chi_{20}(x_{20})  \\
&&\chi_{21}(x_{21})  \chi_{17}(-x_{22})  \chi_{23}(x_{23})  \chi_{24}(x_{24})  \chi_{25}(x_{25})  \chi_{26}(x_{26})  \chi_{27}(x_{27})  \chi_{31}(-x_{28})  \chi_{29}(x_{29})  \chi_{33}(-x_{30})  \\
&&\chi_{28}(x_{31})  \chi_{34}(-x_{32})  \chi_{30}(x_{33})  \chi_{32}(x_{34})  \chi_{35}(x_{35})  \chi_{36}(x_{36})   \\
6 & t_{5}x_{6}/t_{6}^2 & h_{6}(c_6/r_6)h\cdot \chi_{i}(-{}(c_6/r_6)^2x_6)\cdot\chi_{1}(x_{1})  \chi_{2}(x_{2})  \chi_{3}(x_{3})  \chi_{4}(x_{4})  \chi_{5}(-x_{5}x_{6}) \chi_{11}(x_{5})  \cdot \chi_{7}(x_{7})  \chi_{8}(x_{8})  \chi_{9}(x_{9})  \chi_{16}(x_{10})  \\
&&\chi_{5}(-x_{11})  \chi_{12}(x_{12})  \chi_{13}(x_{13})  \chi_{20}(x_{14})  \chi_{21}(x_{15})  \chi_{10}(-x_{16})  \chi_{17}(x_{17})  \chi_{23}(x_{18})  \chi_{25}(x_{19})  \chi_{14}(-x_{20})  \\
&&\chi_{15}(-x_{21})  \chi_{27}(x_{22})  \chi_{18}(-x_{23})  \chi_{28}(x_{24})  \chi_{19}(-x_{25})  \chi_{30}(x_{26})  \chi_{22}(-x_{27})  \chi_{24}(-x_{28})  \chi_{32}(x_{29})  \chi_{26}(-x_{30})  \\
&&\chi_{31}(x_{31})  \chi_{29}(-x_{32})  \chi_{33}(x_{33})  \chi_{34}(x_{34})  \chi_{35}(x_{35})  \chi_{36}(x_{36})   \\
\end{array}$
}\medskip\caption{Iwasawa coordinates for $E_6$ with compensating element 
\mbox{$\kappa_i(c_ix_i)=\chi_i(-c_ix_i)h_\beta(-1/\sqrt{(c_ix_i)^2+1})\wt_i \chi_{i}(-c_ix_i)$}}
\end{table}

\begin{landscape}
\begin{table}
{\tiny
$\begin{array}{lll} i & c_i & \kappa_{i}(c_ix_i)hu\wt_i \\\hline
1 & t_{2}/t_{1}^2 & h_{1}(c_1/r_1)h\cdot \chi_{1}(-(c_1/r_1)^2x_1)\cdot\chi_{8}(-x_{2})  \chi_{3}(x_{3})  \chi_{4}(x_{4})  \chi_{5}(x_{5})  \chi_{6}(x_{6})  \chi_{7}(x_{7})  \chi_{2}(x_{8})  \chi_{9}(x_{9})  \chi_{14}(-x_{10})  \chi_{11}(x_{11})  \chi_{12}(x_{12})  \chi_{13}(x_{13})  \chi_{10}(x_{14}) \chi_{20}(-x_{15})  \chi_{16}(x_{16})  \chi_{21}(-x_{17})\\&&
\chi_{18}(x_{18})  \chi_{19}(x_{19})  \chi_{15}(x_{20})  \chi_{17}(x_{21})  \chi_{26}(-x_{22})  \chi_{23}(x_{23})  \chi_{27}(-x_{24})  \chi_{25}(x_{25})  \chi_{22}(x_{26})  \chi_{24}(x_{27})  \chi_{32}(-x_{28})  \chi_{33}(-x_{29})  \chi_{30}(x_{30}) \\&& \chi_{34}(-x_{31})  \chi_{28}(x_{32})  \chi_{29}(x_{33})  \chi_{31}(x_{34})  \chi_{38}(-x_{35})  \chi_{39}(-x_{36})  \chi_{37}(x_{37}) \chi_{35}(x_{38})  \chi_{36}(x_{39})  \chi_{43}(-x_{40})  \chi_{44}(-x_{41})  \chi_{42}(x_{42})  \chi_{40}(x_{43})\\&&  \chi_{41}(x_{44})  \chi_{48}(-x_{45})  \chi_{46}(x_{46})  \chi_{47}(x_{47})  \chi_{45}(x_{48})  \chi_{52}(-x_{49})  \chi_{50}(x_{50})  \chi_{51}(x_{51})  \chi_{49}(x_{52})  \chi_{53}(x_{53})  \chi_{54}(x_{54})  \chi_{55}(x_{55})  \chi_{56}(x_{56}) \\&& \chi_{57}(x_{57})  \chi_{58}(x_{58})  \chi_{59}(x_{59})  \chi_{60}(x_{60})  \chi_{61}(x_{61})  \chi_{63}(-x_{62})  \chi_{62}(x_{63})   \\
2 & t_{1}t_{4}/t_{2}^2 & h_{2}(c_2/r_2)h\cdot \chi_{2}(-(c_2/r_2)^2x_2)\cdot\chi_{1}(-x_{1}x_{2}) \chi_{8}(x_{1})  \cdot \chi_{3}(x_{3})  \chi_{10}(-x_{4})  \chi_{5}(x_{5})  \chi_{6}(x_{6})  \chi_{7}(x_{7})  \chi_{1}(-x_{8})  \chi_{15}(-x_{9})  \chi_{4}(x_{10})  \chi_{17}(-x_{11})  \chi_{12}(x_{12})  \chi_{13}(x_{13})  \chi_{14}(x_{14})  \chi_{9}(x_{15})\\&&  \chi_{22}(-x_{16})  \chi_{11}(x_{17})  \chi_{24}(-x_{18})  \chi_{19}(x_{19})  \chi_{20}(x_{20})  \chi_{21}(x_{21})  \chi_{16}(x_{22})  \chi_{29}(-x_{23})  \chi_{18}(x_{24})  \chi_{31}(-x_{25}) \\&& \chi_{26}(x_{26})  \chi_{27}(x_{27})  \chi_{28}(x_{28})  \chi_{23}(x_{29})  \chi_{36}(-x_{30})  \chi_{25}(x_{31})  \chi_{37}(-x_{32})  \chi_{33}(x_{33})  \chi_{34}(x_{34})  \chi_{35}(x_{35})  \chi_{30}(x_{36})  \chi_{32}(x_{37})  \chi_{42}(-x_{38})  \chi_{39}(x_{39})  \chi_{40}(x_{40})  \chi_{41}(x_{41}) \\&& \chi_{38}(x_{42})  \chi_{46}(-x_{43})  \chi_{47}(-x_{44})  \chi_{45}(x_{45})  \chi_{43}(x_{46})  \chi_{44}(x_{47})  \chi_{51}(-x_{48})  \chi_{49}(x_{49})  \chi_{50}(x_{50})  \chi_{48}(x_{51})  \chi_{55}(-x_{52})  \chi_{53}(x_{53})  \chi_{54}(x_{54})  \chi_{52}(x_{55})  \chi_{56}(x_{56})  \chi_{57}(x_{57}) \\&& \chi_{58}(x_{58})  \chi_{59}(x_{59})  \chi_{60}(x_{60})  \chi_{62}(-x_{61})  \chi_{61}(x_{62})  \chi_{63}(x_{63})   \\
3 & t_{4}/t_{3}^2 & h_{3}(c_3/r_3)h\cdot \chi_{3}(-(c_3/r_3)^2x_3)\cdot\chi_{1}(x_{1})  \chi_{2}(x_{2})  \cdot \chi_{9}(-x_{4})  \chi_{5}(x_{5})  \chi_{6}(x_{6})  \chi_{7}(x_{7})  \chi_{8}(x_{8})  \chi_{4}(x_{9})  \chi_{15}(-x_{10})  \chi_{16}(-x_{11})  \chi_{12}(x_{12})  \chi_{13}(x_{13})  \chi_{20}(-x_{14})  \chi_{10}(x_{15})  \chi_{11}(x_{16})  \chi_{22}(-x_{17}) \\&& \chi_{23}(-x_{18})  \chi_{19}(x_{19})  \chi_{14}(x_{20})  \chi_{26}(-x_{21})  \chi_{17}(x_{22})  \chi_{18}(x_{23})  \chi_{29}(-x_{24})  \chi_{30}(-x_{25})  \chi_{21}(x_{26})  \chi_{33}(-x_{27})  \chi_{28}(x_{28})  \chi_{24}(x_{29})  \chi_{25}(x_{30})  \chi_{36}(-x_{31})  \chi_{32}(x_{32})  \chi_{27}(x_{33})  \chi_{39}(-x_{34})\\&&  \chi_{35}(x_{35})  \chi_{31}(x_{36})  \chi_{37}(x_{37})  \chi_{38}(x_{38})  \chi_{34}(x_{39})  \chi_{40}(x_{40})  \chi_{41}(x_{41})  \chi_{42}(x_{42})  \chi_{43}(x_{43})  \chi_{44}(x_{44})  \chi_{45}(x_{45})  \chi_{46}(x_{46})  \chi_{47}(x_{47})  \chi_{48}(x_{48})  \chi_{49}(x_{49})  \chi_{53}(-x_{50})\\&&  \chi_{51}(x_{51})  \chi_{52}(x_{52})  \chi_{50}(x_{53})  \chi_{56}(-x_{54})  \chi_{55}(x_{55})  \chi_{54}(x_{56})  \chi_{58}(-x_{57})  \chi_{57}(x_{58})  \chi_{60}(-x_{59})  \chi_{59}(x_{60})  \chi_{61}(x_{61})  \chi_{62}(x_{62})  \chi_{63}(x_{63})   \\
4 & t_{2}t_{3}t_{5}/t_{4}^2 & h_{4}(c_4/r_4)h\cdot \chi_{4}(-(c_4/r_4)^2x_4)\cdot\chi_{1}(x_{1})  \chi_{2}(-x_{2}x_{4}) \chi_{10}(x_{2})  \chi_{3}(-x_{3}x_{4}) \chi_{9}(x_{3})  \cdot \chi_{11}(-x_{5})  \chi_{6}(x_{6})  \chi_{7}(x_{7})  \chi_{14}(x_{8})  \chi_{3}(-x_{9})  \chi_{2}(-x_{10})  \chi_{5}(x_{11})  \chi_{18}(-x_{12})  \chi_{13}(x_{13})  \chi_{8}(-x_{14}) \\&& \chi_{15}(x_{15})  \chi_{16}(x_{16})  \chi_{17}(x_{17})  \chi_{12}(x_{18})  \chi_{25}(-x_{19})  \chi_{20}(x_{20})  \chi_{21}(x_{21})  \chi_{28}(-x_{22})  \chi_{23}(x_{23})  \chi_{24}(x_{24})  \chi_{19}(x_{25})  \chi_{32}(-x_{26})  \chi_{27}(x_{27})  \chi_{22}(x_{28})  \chi_{35}(-x_{29})  \chi_{30}(x_{30}) \\&& \chi_{31}(x_{31})  \chi_{26}(x_{32})  \chi_{38}(-x_{33})  \chi_{34}(x_{34})  \chi_{29}(x_{35})  \chi_{41}(-x_{36})  \chi_{37}(x_{37})  \chi_{33}(x_{38})  \chi_{44}(-x_{39})  \chi_{40}(x_{40})  \chi_{36}(x_{41})  \chi_{42}(x_{42})  \chi_{43}(x_{43})  \chi_{39}(x_{44})  \chi_{45}(x_{45})  \chi_{50}(-x_{46}) \\&& \chi_{47}(x_{47})  \chi_{48}(x_{48})  \chi_{49}(x_{49})  \chi_{46}(x_{50})  \chi_{54}(-x_{51})  \chi_{52}(x_{52})  \chi_{53}(x_{53})  \chi_{51}(x_{54})  \chi_{57}(-x_{55})  \chi_{56}(x_{56}) \\&& \chi_{55}(x_{57})  \chi_{58}(x_{58})  \chi_{59}(x_{59})  \chi_{61}(-x_{60})  \chi_{60}(x_{61})  \chi_{62}(x_{62})  \chi_{63}(x_{63})   \\
5 & t_{4}t_{6}/t_{5}^2 & h_{5}(c_5/r_5)h\cdot \chi_{5}(-(c_5/r_5)^2x_5)\cdot\chi_{1}(x_{1})  \chi_{2}(x_{2})  \chi_{3}(x_{3})  \chi_{4}(-x_{4}x_{5}) \chi_{11}(x_{4})  \cdot \chi_{12}(-x_{6})  \chi_{7}(x_{7})  \chi_{8}(x_{8})  \chi_{16}(x_{9})  \chi_{17}(x_{10})  \chi_{4}(-x_{11})  \chi_{6}(x_{12})  \chi_{19}(-x_{13})  \chi_{21}(x_{14})  \chi_{22}(x_{15})  \chi_{9}(-x_{16})\\&&  \chi_{10}(-x_{17})  \chi_{18}(x_{18})  \chi_{13}(x_{19})  \chi_{26}(x_{20})  \chi_{14}(-x_{21})  \chi_{15}(-x_{22})  \chi_{23}(x_{23})  \chi_{24}(x_{24})  \chi_{25}(x_{25})  \chi_{20}(-x_{26})  \chi_{27}(x_{27})  \chi_{28}(x_{28})  \chi_{29}(x_{29})  \chi_{30}(x_{30})  \chi_{31}(x_{31})  \chi_{32}(x_{32}) \\&& \chi_{33}(x_{33})  \chi_{34}(x_{34})  \chi_{40}(-x_{35})  \chi_{36}(x_{36})  \chi_{37}(x_{37})  \chi_{43}(-x_{38})  \chi_{39}(x_{39})  \chi_{35}(x_{40})  \chi_{45}(-x_{41})  \chi_{46}(-x_{42})  \chi_{38}(x_{43})  \chi_{48}(-x_{44})  \chi_{41}(x_{45})  \chi_{42}(x_{46})  \chi_{51}(-x_{47})  \chi_{44}(x_{48}) \\&& \chi_{49}(x_{49})  \chi_{50}(x_{50})  \chi_{47}(x_{51})  \chi_{52}(x_{52})  \chi_{53}(x_{53})  \chi_{54}(x_{54})  \chi_{55}(x_{55})  \chi_{56}(x_{56})  \chi_{59}(-x_{57})  \chi_{60}(-x_{58})  \chi_{57}(x_{59})  \chi_{58}(x_{60})  \chi_{61}(x_{61})  \chi_{62}(x_{62})  \chi_{63}(x_{63})   \\
6 & t_{5}t_{7}/t_{6}^2 & h_{6}(c_6/r_6)h\cdot \chi_{6}(-(c_6/r_6)^2x_6)\cdot\chi_{1}(x_{1})  \chi_{2}(x_{2})  \chi_{3}(x_{3})  \chi_{4}(x_{4})  \chi_{5}(-x_{5}x_{6}) \chi_{12}(x_{5})  \cdot \chi_{13}(-x_{7})  \chi_{8}(x_{8})  \chi_{9}(x_{9})  \chi_{10}(x_{10})  \chi_{18}(x_{11})  \chi_{5}(-x_{12})  \chi_{7}(x_{13})  \chi_{14}(x_{14})  \chi_{15}(x_{15})  \chi_{23}(x_{16})\\&&  \chi_{24}(x_{17})  \chi_{11}(-x_{18})  \chi_{19}(x_{19})  \chi_{20}(x_{20})  \chi_{27}(x_{21})  \chi_{29}(x_{22})  \chi_{16}(-x_{23})  \chi_{17}(-x_{24})  \chi_{25}(x_{25})  \chi_{33}(x_{26})  \chi_{21}(-x_{27})  \chi_{35}(x_{28})  \chi_{22}(-x_{29})  \chi_{30}(x_{30})  \chi_{31}(x_{31})  \chi_{38}(x_{32})\\&&  \chi_{26}(-x_{33})  \chi_{34}(x_{34})  \chi_{28}(-x_{35})  \chi_{36}(x_{36})  \chi_{42}(x_{37})  \chi_{32}(-x_{38})  \chi_{39}(x_{39})  \chi_{40}(x_{40})  \chi_{41}(x_{41})  \chi_{37}(-x_{42})  \chi_{43}(x_{43})  \chi_{44}(x_{44})  \chi_{49}(-x_{45})  \chi_{46}(x_{46})  \chi_{47}(x_{47})  \chi_{52}(-x_{48}) \\&& \chi_{45}(x_{49})  \chi_{50}(x_{50})  \chi_{55}(-x_{51})  \chi_{48}(x_{52})  \chi_{53}(x_{53})  \chi_{57}(-x_{54})  \chi_{51}(x_{55})  \chi_{58}(-x_{56})  \chi_{54}(x_{57})  \chi_{56}(x_{58})  \chi_{59}(x_{59})  \chi_{60}(x_{60})  \chi_{61}(x_{61})  \chi_{62}(x_{62})  \chi_{63}(x_{63})   \\
7 & t_{6}/t_{7}^2 & h_{7}(c_7/r_7)h\cdot \chi_{7}(-(c_7/r_7)^2x_7)\cdot\chi_{1}(x_{1})  \chi_{2}(x_{2})  \chi_{3}(x_{3})  \chi_{4}(x_{4})  \chi_{5}(x_{5})  \chi_{6}(-x_{6}x_{7}) \chi_{13}(x_{6})  \cdot \chi_{8}(x_{8})  \chi_{9}(x_{9})  \chi_{10}(x_{10})  \chi_{11}(x_{11})  \chi_{19}(x_{12})  \chi_{6}(-x_{13})  \chi_{14}(x_{14})  \chi_{15}(x_{15})  \chi_{16}(x_{16})\\&&  \chi_{17}(x_{17})  \chi_{25}(x_{18})  \chi_{12}(-x_{19})  \chi_{20}(x_{20})  \chi_{21}(x_{21})  \chi_{22}(x_{22})  \chi_{30}(x_{23})  \chi_{31}(x_{24})  \chi_{18}(-x_{25})  \chi_{26}(x_{26})  \chi_{34}(x_{27})  \chi_{28}(x_{28})  \chi_{36}(x_{29})  \chi_{23}(-x_{30})  \chi_{24}(-x_{31})  \chi_{32}(x_{32}) \\&& \chi_{39}(x_{33})  \chi_{27}(-x_{34})  \chi_{41}(x_{35})  \chi_{29}(-x_{36})  \chi_{37}(x_{37})  \chi_{44}(x_{38})  \chi_{33}(-x_{39})  \chi_{45}(x_{40})  \chi_{35}(-x_{41})  \chi_{47}(x_{42})  \chi_{48}(x_{43})  \chi_{38}(-x_{44})  \chi_{40}(-x_{45})  \chi_{51}(x_{46})  \chi_{42}(-x_{47})  \chi_{43}(-x_{48})\\&&  \chi_{49}(x_{49})  \chi_{54}(x_{50})  \chi_{46}(-x_{51})  \chi_{52}(x_{52})  \chi_{56}(x_{53})  \chi_{50}(-x_{54})  \chi_{55}(x_{55})  \chi_{53}(-x_{56})  \chi_{57}(x_{57})  \chi_{58}(x_{58})  \chi_{59}(x_{59})  \chi_{60}(x_{60})  \chi_{61}(x_{61})  \chi_{62}(x_{62})  \chi_{63}(x_{63})   
\end{array}$
}
\caption{Iwasawa coordinates for $E_7$ with compensating element 
\mbox{$\kappa_i(c_ix_i)=\chi_i(-c_ix_i)h_i(-1/\sqrt{(c_ix_i)^2+1})\wt_i \chi_{i}(-c_ix_i)$}}
\end{table}

\begin{table}
{\tiny
$\begin{array}{lll} i & c_i & \kappa_i{}(c_ix_i)hu\wt_i \\\hline
1 & t_{2}/t_{1}^2 & h_{1}(c_1/r_1)h\cdot \chi_{1}(-(c_1/r_1)^2x_1)\cdot\chi_{9}(-x_{2})  \chi_{3}(x_{3})  \chi_{4}(x_{4})  \chi_{5}(x_{5})  \chi_{6}(x_{6})  \chi_{7}(x_{7})  \chi_{8}(x_{8})  \chi_{2}(x_{9})  \chi_{10}(x_{10})  \chi_{16}(-x_{11})  \chi_{12}(x_{12})  \chi_{13}(x_{13})  \chi_{14}(x_{14})  \chi_{15}(x_{15})  \chi_{11}(x_{16})  \chi_{23}(-x_{17})  \\&&
\chi_{18}(x_{18})  \chi_{24}(-x_{19})  \chi_{20}(x_{20})  \chi_{21}(x_{21})  \chi_{22}(x_{22})  \chi_{17}(x_{23})  \chi_{19}(x_{24})  \chi_{30}(-x_{25})  \chi_{26}(x_{26})  \chi_{31}(-x_{27})  \chi_{28}(x_{28})  \chi_{29}(x_{29})  \chi_{25}(x_{30})  \chi_{27}(x_{31})  \chi_{37}(-x_{32})  \chi_{38}(-x_{33}) \\&&
\chi_{34}(x_{34})  \chi_{39}(-x_{35})  \chi_{36}(x_{36})  \chi_{32}(x_{37})  \chi_{33}(x_{38})  \chi_{35}(x_{39})  \chi_{45}(-x_{40})  \chi_{46}(-x_{41})  \chi_{42}(x_{42})  \chi_{47}(-x_{43})  \chi_{44}(x_{44})  \chi_{40}(x_{45})  \chi_{41}(x_{46})  \chi_{43}(x_{47})  \chi_{52}(-x_{48})  \chi_{53}(-x_{49}) \\&& 
\chi_{54}(-x_{50})  \chi_{51}(x_{51})  \chi_{48}(x_{52})  \chi_{49}(x_{53})  \chi_{50}(x_{54})  \chi_{59}(-x_{55})  \chi_{60}(-x_{56})  \chi_{57}(x_{57})  \chi_{58}(x_{58})  \chi_{55}(x_{59})  \chi_{56}(x_{60})  \chi_{66}(-x_{61})  \chi_{67}(-x_{62})  \chi_{63}(x_{63})  \chi_{64}(x_{64})  \chi_{65}(x_{65}) \\&&
\chi_{61}(x_{66})  \chi_{62}(x_{67})  \chi_{73}(-x_{68})  \chi_{69}(x_{69})  \chi_{70}(x_{70})  \chi_{71}(x_{71})  \chi_{72}(x_{72})  \chi_{68}(x_{73})  \chi_{79}(-x_{74})  \chi_{75}(x_{75})  \chi_{76}(x_{76})  \chi_{77}(x_{77})  \chi_{78}(x_{78})  \chi_{74}(x_{79})  \chi_{80}(x_{80})  \chi_{81}(x_{81}) \\&&
\chi_{82}(x_{82})  \chi_{83}(x_{83})  \chi_{84}(x_{84})  \chi_{85}(x_{85})  \chi_{86}(x_{86})  \chi_{87}(x_{87})  \chi_{88}(x_{88})  \chi_{89}(x_{89})  \chi_{90}(x_{90})  \chi_{91}(x_{91})  \chi_{92}(x_{92})  \chi_{97}(-x_{93})  \chi_{94}(x_{94})  \chi_{95}(x_{95})  \chi_{96}(x_{96})  \chi_{93}(x_{97})  \\&&
\chi_{101}(-x_{98})  \chi_{99}(x_{99})  \chi_{100}(x_{100})  \chi_{98}(x_{101})  \chi_{104}(-x_{102})  \chi_{103}(x_{103})  \chi_{102}(x_{104})  \chi_{107}(-x_{105})  \chi_{106}(x_{106})  \chi_{105}(x_{107})  \chi_{109}(-x_{108})  \chi_{108}(x_{109})  \chi_{111}(-x_{110})  \chi_{110}(x_{111})  \chi_{113}(-x_{112})\\&&
\chi_{112}(x_{113})  \chi_{114}(x_{114})  \chi_{115}(x_{115})  \chi_{116}(x_{116})  \chi_{117}(x_{117})  \chi_{118}(x_{118})  \chi_{119}(x_{119})  \chi_{120}(x_{120})   \\
2 & t_{1}t_{4}/t_{2}^2 & h_{2}(c_2/r_2)h\cdot \chi_{2}(-(c_2/r_2)^2x_2)\cdot\chi_{1}(-x_{1}x_{2}) \chi_{9}(x_{1})  \cdot \chi_{3}(x_{3})  \chi_{11}(-x_{4})  \chi_{5}(x_{5})  \chi_{6}(x_{6})  \chi_{7}(x_{7})  \chi_{8}(x_{8})  \chi_{1}(-x_{9})  \chi_{17}(-x_{10})  \chi_{4}(x_{11})  \chi_{19}(-x_{12})  \chi_{13}(x_{13})  \chi_{14}(x_{14})  \chi_{15}(x_{15}) \\&&
\chi_{16}(x_{16})  \chi_{10}(x_{17})  \chi_{25}(-x_{18})  \chi_{12}(x_{19})  \chi_{27}(-x_{20})  \chi_{21}(x_{21})  \chi_{22}(x_{22})  \chi_{23}(x_{23})  \chi_{24}(x_{24})  \chi_{18}(x_{25})  \chi_{33}(-x_{26})  \chi_{20}(x_{27})  \chi_{35}(-x_{28})  \chi_{29}(x_{29})  \chi_{30}(x_{30})  \chi_{31}(x_{31}) \\&&
\chi_{32}(x_{32})  \chi_{26}(x_{33})  \chi_{41}(-x_{34})  \chi_{28}(x_{35})  \chi_{43}(-x_{36})  \chi_{44}(-x_{37})  \chi_{38}(x_{38})  \chi_{39}(x_{39})  \chi_{40}(x_{40})  \chi_{34}(x_{41})  \chi_{50}(-x_{42})  \chi_{36}(x_{43})  \chi_{37}(x_{44})  \chi_{51}(-x_{45})  \chi_{46}(x_{46})  \chi_{47}(x_{47})  \\&&
\chi_{48}(x_{48})  \chi_{49}(x_{49})  \chi_{42}(x_{50})  \chi_{45}(x_{51})  \chi_{57}(-x_{52})  \chi_{58}(-x_{53})  \chi_{54}(x_{54})  \chi_{55}(x_{55})  \chi_{56}(x_{56})  \chi_{52}(x_{57})  \chi_{53}(x_{58})  \chi_{64}(-x_{59})  \chi_{65}(-x_{60})  \chi_{61}(x_{61})  \chi_{62}(x_{62})  \chi_{63}(x_{63})  \\&&
\chi_{59}(x_{64})  \chi_{60}(x_{65})  \chi_{71}(-x_{66})  \chi_{72}(-x_{67})  \chi_{68}(x_{68})  \chi_{69}(x_{69})  \chi_{70}(x_{70})  \chi_{66}(x_{71})  \chi_{67}(x_{72})  \chi_{78}(-x_{73})  \chi_{74}(x_{74})  \chi_{75}(x_{75})  \chi_{76}(x_{76})  \chi_{77}(x_{77})  \chi_{73}(x_{78})  \chi_{84}(-x_{79}) \\&& 
\chi_{80}(x_{80})  \chi_{81}(x_{81})  \chi_{82}(x_{82})  \chi_{83}(x_{83})  \chi_{79}(x_{84})  \chi_{85}(x_{85})  \chi_{86}(x_{86})  \chi_{87}(x_{87})  \chi_{88}(x_{88})  \chi_{93}(-x_{89})  \chi_{90}(x_{90})  \chi_{91}(x_{91})  \chi_{92}(x_{92})  \chi_{89}(x_{93})  \chi_{98}(-x_{94})  \chi_{95}(x_{95})  \\&&
\chi_{96}(x_{96})  \chi_{97}(x_{97})  \chi_{94}(x_{98})  \chi_{102}(-x_{99})  \chi_{100}(x_{100})  \chi_{101}(x_{101})  \chi_{99}(x_{102})  \chi_{105}(-x_{103})  \chi_{104}(x_{104})  \chi_{103}(x_{105})  \chi_{108}(-x_{106})  \chi_{107}(x_{107})  \chi_{106}(x_{108})  \chi_{109}(x_{109})  \chi_{110}(x_{110})  \chi_{114}(-x_{111})  \\&&
\chi_{112}(x_{112})  \chi_{115}(-x_{113})  \chi_{111}(x_{114})  \chi_{113}(x_{115})  \chi_{116}(x_{116})  \chi_{117}(x_{117})  \chi_{118}(x_{118})  \chi_{119}(x_{119})  \chi_{120}(x_{120})   \\
3 & t_{4}/t_{3}^2 & h_{3}(c_3/r_3)h\cdot \chi_{3}(-(c_3/r_3)^2x_3)\cdot\chi_{1}(x_{1})  \chi_{2}(x_{2})  \cdot \chi_{10}(-x_{4})  \chi_{5}(x_{5})  \chi_{6}(x_{6})  \chi_{7}(x_{7})  \chi_{8}(x_{8})  \chi_{9}(x_{9})  \chi_{4}(x_{10})  \chi_{17}(-x_{11})  \chi_{18}(-x_{12})  \chi_{13}(x_{13})  \chi_{14}(x_{14})  \chi_{15}(x_{15})  \chi_{23}(-x_{16})  \chi_{11}(x_{17})  \\&&
\chi_{12}(x_{18})  \chi_{25}(-x_{19})  \chi_{26}(-x_{20})  \chi_{21}(x_{21})  \chi_{22}(x_{22})  \chi_{16}(x_{23})  \chi_{30}(-x_{24})  \chi_{19}(x_{25})  \chi_{20}(x_{26})  \chi_{33}(-x_{27})  \chi_{34}(-x_{28})  \chi_{29}(x_{29})  \chi_{24}(x_{30})  \chi_{38}(-x_{31})  \chi_{32}(x_{32})  \chi_{27}(x_{33})  \\&&
\chi_{28}(x_{34})  \chi_{41}(-x_{35})  \chi_{42}(-x_{36})  \chi_{37}(x_{37})  \chi_{31}(x_{38})  \chi_{46}(-x_{39})  \chi_{40}(x_{40})  \chi_{35}(x_{41})  \chi_{36}(x_{42})  \chi_{50}(-x_{43})  \chi_{44}(x_{44})  \chi_{45}(x_{45})  \chi_{39}(x_{46})  \chi_{54}(-x_{47})  \chi_{48}(x_{48})  \chi_{49}(x_{49})  \\&&
\chi_{43}(x_{50})  \chi_{51}(x_{51})  \chi_{52}(x_{52})  \chi_{53}(x_{53})  \chi_{47}(x_{54})  \chi_{55}(x_{55})  \chi_{56}(x_{56})  \chi_{57}(x_{57})  \chi_{58}(x_{58})  \chi_{59}(x_{59})  \chi_{60}(x_{60})  \chi_{61}(x_{61})  \chi_{62}(x_{62})  \chi_{69}(-x_{63})  \chi_{64}(x_{64})  \chi_{65}(x_{65}) \\&&
 \chi_{66}(x_{66})  \chi_{67}(x_{67})  \chi_{68}(x_{68})  \chi_{63}(x_{69})  \chi_{75}(-x_{70})  \chi_{71}(x_{71})  \chi_{72}(x_{72})  \chi_{73}(x_{73})  \chi_{74}(x_{74})  \chi_{70}(x_{75})  \chi_{80}(-x_{76})  \chi_{81}(-x_{77})  \chi_{78}(x_{78})  \chi_{79}(x_{79})  \chi_{76}(x_{80})  \chi_{77}(x_{81})  \\&&
 \chi_{85}(-x_{82})  \chi_{86}(-x_{83})  \chi_{84}(x_{84})  \chi_{82}(x_{85})  \chi_{83}(x_{86})  \chi_{90}(-x_{87})  \chi_{91}(-x_{88})  \chi_{89}(x_{89})  \chi_{87}(x_{90})  \chi_{88}(x_{91})  \chi_{95}(-x_{92})  \chi_{93}(x_{93})  \chi_{94}(x_{94})  \chi_{92}(x_{95})  \chi_{100}(-x_{96})  \chi_{97}(x_{97})  \\&&
 \chi_{98}(x_{98})  \chi_{99}(x_{99})  \chi_{96}(x_{100})  \chi_{101}(x_{101})  \chi_{102}(x_{102})  \chi_{103}(x_{103})  \chi_{104}(x_{104})  \chi_{105}(x_{105})  \chi_{106}(x_{106})  \chi_{107}(x_{107})  \chi_{108}(x_{108})  \chi_{109}(x_{109})  \chi_{112}(-x_{110})  \chi_{113}(-x_{111})  \chi_{110}(x_{112})  \\&&
 \chi_{111}(x_{113})  \chi_{115}(-x_{114})  \chi_{114}(x_{115})  \chi_{116}(x_{116})  \chi_{117}(x_{117})  \chi_{118}(x_{118})  \chi_{119}(x_{119})  \chi_{120}(x_{120})   \\
4 & t_{2}t_{3}t_{5}/t_{4}^2 & h_{4}(c_4/r_4)h\cdot \chi_{4}(-(c_4/r_4)^2x_4)\cdot\chi_{1}(x_{1})  \chi_{2}(-x_{2}x_{4}) \chi_{11}(x_{2})  \chi_{3}(-x_{3}x_{4}) \chi_{10}(x_{3})  \cdot \chi_{12}(-x_{5})  \chi_{6}(x_{6})  \chi_{7}(x_{7})  \chi_{8}(x_{8})  \chi_{16}(x_{9})  \chi_{3}(-x_{10})  \chi_{2}(-x_{11})  \chi_{5}(x_{12})  \chi_{20}(-x_{13})  \chi_{14}(x_{14})\\&&
\chi_{15}(x_{15})  \chi_{9}(-x_{16})  \chi_{17}(x_{17})  \chi_{18}(x_{18})  \chi_{19}(x_{19})  \chi_{13}(x_{20})  \chi_{28}(-x_{21})  \chi_{22}(x_{22})  \chi_{23}(x_{23})  \chi_{24}(x_{24})  \chi_{32}(-x_{25})  \chi_{26}(x_{26})  \chi_{27}(x_{27})  \chi_{21}(x_{28})  \chi_{36}(-x_{29})\\&&
 \chi_{37}(-x_{30})  \chi_{31}(x_{31})  \chi_{25}(x_{32})  \chi_{40}(-x_{33})  \chi_{34}(x_{34})  \chi_{35}(x_{35})  \chi_{29}(x_{36})  \chi_{30}(x_{37})  \chi_{45}(-x_{38})  \chi_{39}(x_{39})  \chi_{33}(x_{40})  \chi_{49}(-x_{4
1})  \chi_{42}(x_{42})  \chi_{43}(x_{43})  \chi_{44}(x_{44})\\&&
\chi_{38}(x_{45})  \chi_{53}(-x_{46})  \chi_{47}(x_{47})  \chi_{48}(x_{48})  \chi_{41}(x_{49})  \chi_{56}(-x_{50})  \chi_{51}(x_{51})  \chi_{52}(x_{52})  \chi_{46}(x_{53})  \chi_{60}(-x_{54})  \chi_{55}(x_{55})  \chi_{50}(x_{56})  \chi_{63}(-x_{57})  \chi_{58}(x_{58})  \chi_{59}(x_{59})  \chi_{54}(x_{60})  \\&&
\chi_{61}(x_{61})  \chi_{62}(x_{62})  \chi_{57}(x_{63})  \chi_{70}(-x_{64})  \chi_{65}(
x_{65})  \chi_{66}(x_{66})  \chi_{67}(x_{67})  \chi_{68}(x_{68})  \chi_{69}(x_{69})  \chi_{64}(x_{70})  \chi_{76}(-x_{71})  \chi_{77}(-x_{72})  \chi_{73}(x_{73})  \chi_{74}(x_{74})  \\&&
\chi_{75}(x_{75})  \chi_{71}(x_{76})  \chi_{72}(x_{77})  \chi_{83}(-x_{78})  \chi_{79}(x_{79})  \chi_{80}(x_{80})  \chi_{81}(x_{81})  \chi_{82}(x_{82})  \chi_{78}(x_{83})  \chi_{88}(-x_{84})  \chi_{89}(-x_{85})  \chi_{86}(x_{86})  \chi_{87}(x_{87})  \chi_{84}(x_{88}) \chi_{85}(x_{89})  \chi_{94}(-x_{90})  \\&&
\chi_{91}(x_{91})  \chi_{92}(x_{92})  \chi_{93}(x_{93})  \chi_{90}(x_{94})  \chi_{99}(-x_{95})  \chi_{96}(x_{96})  \chi_{97}(x_{97})  \chi_{98}(x_{98})  \chi_{95}(x_{99})  \chi_{103}(-x_{100})  \chi_{101}(x_{101})  \chi_{102}(x_{102})  \chi_{100}(x_{103})  \chi_{104}(x_{104})  \chi_{105}(x_{105})  \chi_{106}(x_{106})\\&&
  \chi_{107}(x_{107})  \chi_{110}(-x_{108})  \chi_{111}(-x_{109})  \chi_{108}(x_{110}) \chi_{109}(x_{111})  \chi_{112}(x_{112})  \chi_{113}(x_{113})  \chi_{114}(x_{114})  \chi_{116}(-x_{115})  \chi_{115}(x_{116})  \chi_{117}(x_{117})  \chi_{118}(x_{118})  \chi_{119}(x_{119})  \chi_{120}(x_{120})   \\
\end{array}$
}
\caption{Iwasawa coordinates for $E_8$ with compensating element 
\mbox{$\kappa_i(c_ix_i)=\chi_i(-c_ix_i)h_i(-1/\sqrt{(c_ix_i)^2+1})\wt_i \chi_{i}(-c_ix_i)$: Part 1}}\label{T-last}
\end{table}
\begin{table}
{\tiny
$\begin{array}{lll} i & c_i & \kappa_i{}(c_ix_i)hu\wt_i \\\hline
5 & t_{4}t_{6}/t_{5}^2 & h_{5}(c_5/r_5)h\cdot \chi_{5}(-(c_5/r_5)^2x_5)\cdot\chi_{1}(x_{1})  \chi_{2}(x_{2})  \chi_{3}(x_{3})  \chi_{4}(-x_{4}x_{5}) \chi_{12}(x_{4})  \cdot \chi_{13}(-x_{6})  \chi_{7}(x_{7})  \chi_{8}(x_{8})  \chi_{9}(x_{9})  \chi_{18}(x_{10})  \chi_{19}(x_{11})  \chi_{4}(-x_{12})  \chi_{6}(x_{13})  \chi_{21}(-x_{14})  \chi_{15}(x_{15})  \chi_{24}(x_{16})\\&&
\chi_{25}(x_{17})  \chi_{10}(-x_{18})  \chi_{11}(-x_{19}) \chi_{20}(x_{20})  \chi_{14}(x_{21})  \chi_{29}(-x_{22})  \chi_{30}(x_{23})  \chi_{16}(-x_{24})  \chi_{17}(-x_{25})  \chi_{26}(x_{26})  \chi_{27}(x_{27})  \chi_{28}(x_{28})  \chi_{22}(x_{29})  \chi_{23}(-x_{30})  \chi_{31}(x_{31})  \chi_{32}(x_{32}) \\&&
\chi_{33}(x_{33})  \chi_{34}(x_{34})  \chi_{35}(x_{35})  \chi_{36}(x_{36})  \chi_{37}(x_{37})  \chi_{38}(x_{38})  \chi_{39}(x_{39})  \chi_{48}(-x_{40})  \chi_{41}(x_{41})  \chi_{42}(x_{42})  \chi_{43}(x_{43})  \chi_{44}(x_{44})  \chi_{52}(-x_{45})  \chi_{46}(x_{46})  \chi_{47}(x_{47})  \chi_{40}(x_{48})  \\&&
\chi_{55}(-x_{49})  \chi_{50}(x_{50})  \chi_{57}(-x_{51})  \chi_{45}(x_{52})  \chi_{59}(-x_{53})  \chi_{54}(x_{54})  \chi_{49}(
x_{55})  \chi_{62}(-x_{56})  \chi_{51}(x_{57})  \chi_{64}(-x_{58})  \chi_{53}(x_{59})  \chi_{67}(-x_{60})  \chi_{61}(x_{61})  \chi_{56}(x_{62})  \chi_{63}(x_{63})  \chi_{58}(x_{64})  \\&&
\chi_{72}(-x_{65})  \chi_{66}(x_{66})\chi_{60}(x_{67})  \chi_{68}(x_{68})  \chi_{69}(x_{69})  \chi_{70}(x_{70})  \chi_{71}(x_{71})  \chi_{65}(x_{72})  \chi_{73}(x_{73})  \chi_{74}(x_{74})  \chi_{75}(x_{75})  \chi_{82}(-x_{76})  \chi_{77}(x_{77})  \chi_{78}(x_{78}) \chi_{79}(x_{79})  \chi_{85}(-x_{80}) \\&&
 \chi_{81}(x_{81})  \chi_{76}(x_{82})  \chi_{87}(-x_{83})  \chi_{84}(x_{84})  \chi_{80}(x_{85})  \chi_{90}(-x_{86})  \chi_{83}(x_{87})  \chi_{92}(-x_{88})  \chi_{89}(x_{89})  \chi_{86}(x_{90})  \chi_{95}(-x_{91})  \chi_{88}(x_{92})  \chi_{93}(x_{93})  \chi_{94}(x_{94})  \chi_{91}(x_{95})  \chi_{96}(x_{96})  \\&&
\chi_{97}(x_{97})  \chi_{98}(x_{98})  \chi_{99}(x_{99})  \chi_{100}(x_{100})  \chi_{101}(x_{101})  \chi_{102}(x_{102})  \chi_{106}(-x_{103})  \chi_{104}(x_{104})  \chi_{108}(-x_{105})  \chi_{103}(x_{106})  \chi_{109}(-x_{107})  \chi_{105}(x_{108})  \chi_{107}(x_{109})  \chi_{110}(x_{110})  \chi_{111}(x_{111})  \\&&
\chi_{112}(x_{112})  \chi_{113}(x_{113})  \chi_{114}(x_{114})  \chi_{115}(x_{115})  \chi_{117}(-x_{116})  \chi_{116}(x_{117})  \chi_{118}(x_{118})  \chi_{119}(x_{119})  \chi_{120}(x_{120})   \\
6 & t_{5}t_{7}/t_{6}^2 & h_{6}(c_6/r_6)h\cdot \chi_{6}(-(c_6/r_6)^2x_6)\cdot\chi_{1}(x_{1})  \chi_{2}(x_{2})  \chi_{3}(x_{3})  \chi_{4}(x_{4})  \chi_{5}(-x_{5}x_{6}) \chi_{13}(x_{5})  \cdot \chi_{14}(-x_{7})  \chi_{8}(x_{8})  \chi_{9}(x_{9})  \chi_{10}(x_{10})  \chi_{11}(x_{11})  \chi_{20}(x_{12})  \chi_{5}(-x_{13})  \chi_{7}(x_{14})  \chi_{22}(-x_{15})  \chi_{16}(x_{16}) \\&&
\chi_{17}(x_{17})  \chi_{26}(x_{18})  \chi_{27}(x_{19})  \chi_{12}(-x_{20})  \chi_{21}(x_{21})  \chi_{15}(x_{22})  \chi_{23}(x_{23})  \chi_{31}(x_{24})  \chi_{33}(x_{25})  \chi_{18}(-x_{26})  \chi_{19}(-x_{27})  \chi_{28}(x_{28})  \chi_{29}(x_{29})  \chi_{38}(x_{30})  \chi_{24}(-x_{31})  \chi_{40}(x_{32})  \\&&
\chi_{25}(-x_{33})  \chi_{34}(x_{34})  \chi_{35}(x_{35})  \chi_{36}(x_{36})  \chi_{45}(x_{37})  \chi_{30}(-x_{38})  \chi_{39}(x_{39})  \chi_{32}(-x_{40})  \chi_{41}(x_{41})  \chi_{42}(x_{42})  \chi_{43}(x_{43})  \chi_{51}(x_{44})  \chi_{37}(-x_{45})  \chi_{46}(x_{46})  \chi_{47}(x_{47})  \chi_{48}(x_{48})  \\&&
\chi_{49}(x_{49})  \chi_{50}(x_{50})  \chi_{44}(-x_{51})  \chi_{52}(x_{52})  \chi_{53}(x_{53})  \chi_{54}(x_{54})  \chi_{61}(-x_{55})  \chi_{56}(x_{56})  \chi_{57}(x_{57})  \chi_{58}(x_{58})  \chi_{66}(-x_{59})  \chi_{60}(x_{60})  \chi_{55}(x_{61})  \chi_{68}(-x_{62})  \chi_{63}(x_{63})  \chi_{71}(-x_{64})  \\&&
\chi_{65}(x_{65})  \chi_{59}(x_{66})  \chi_{73}(-x_{67})  \chi_{62}(x_{68})  \chi_{69}(x_{69})  \chi_{76}(-x_{70})  \chi_{64}(x_{71})  \chi_{78}(-x_{72})  \chi_{67}(x_{73})  \chi_{74}(x_{74})  \chi_{80}(-x_{75})  \chi_{70}(x_{76})  \chi_{83}(-x_{77})  \chi_{72}(x_{78})  \chi_{79}(x_{79})  \chi_{75}(x_{80})  \\&&
\chi_{86}(-x_{81})  \chi_{82}(x_{82})  \chi_{77}(x_{83})  \chi_{84}(x_{84})  \chi_{85}(x_{85})  \chi_{81}(x_{86})  \chi_{87}(x_{87})  \chi_{88}(x_{88})  \chi_{89}(x_{89})  \chi_{90}(x_{90})  \chi_{91}(x_{91})  \chi_{96}(-x_{92})  \chi_{93}(x_{93})  \chi_{94}(x_{94})  \chi_{100}(-x_{95})  \chi_{92}(x_{96}) \\&&
\chi_{97}(x_{97})  \chi_{98}(x_{98})  \chi_{103}(-x_{99})  \chi_{95}(x_{100})  \chi_{101}(x_{101})  \chi_{105}(-x_{102})  \chi_{99}(x_{103})  \chi_{107}(-x_{104})  \chi_{102}(x_{105})  \chi_{106}(x_{106})  \chi_{104}(x_{107})  \chi_{108}(x_{108})  \chi_{109}(x_{109})  \chi_{110}(x_{110})  \chi_{111}(x_{111})  \\&&
\chi_{112}(x_{112})  \chi_{113}(x_{113})  \chi_{114}(x_{114})  \chi_{115}(x_{115})  \chi_{116}(x_{116})  \chi_{118}(-x_{117})  \chi_{117}(x_{118})  \chi_{119}(x_{119})  \chi_{120}(x_{120})   \\
7 & t_{6}t_{8}/t_{7}^2 & h_{7}(c_7/r_7)h\cdot \chi_{7}(-(c_7/r_7)^2x_7)\cdot\chi_{1}(x_{1})  \chi_{2}(x_{2})  \chi_{3}(x_{3})  \chi_{4}(x_{4})  \chi_{5}(x_{5})  \chi_{6}(-x_{6}x_{7}) \chi_{14}(x_{6})  \cdot \chi_{15}(-x_{8})  \chi_{9}(x_{9})  \chi_{10}(x_{10})  \chi_{11}(x_{11})  \chi_{12}(x_{12})  \chi_{21}(x_{13})  \chi_{6}(-x_{14})  \chi_{8}(x_{15})  \chi_{16}(x_{16}) \\&& \chi_{17}(x_{17})  \chi_{18}(x_{18})  \chi_{19}(x_{19})  \chi_{28}(x_{20})  \chi_{13}(-x_{21})  \chi_{22}(x_{22})  \chi_{23}(x_{23})  \chi_{24}(x_{24})  \chi_{25}(x_{25})  \chi_{34}(x_{26})  \chi_{35}(x_{27})  \chi_{20}(-x_{28})  \chi_{29}(x_{29})  \chi_{30}(x_{30})  \chi_{39}(x_{31})  \chi_{32}(x_{32})  \\&&
\chi_{41}(x_{33})  \chi_{26}(-x_{34})  \chi_{27}(-x_{35})  \chi_{36}(x_{36})  \chi_{37}(x_{37})  \chi_{46}(x_{38})  \chi_{31}(-x_{39})  \chi_{49}(x_{40})  \chi_{33}(-x_{41})  \chi_{42}(x_{42})  \chi_{43}(x_{43})  \chi_{44}(x_{44})  \chi_{53}(x_{45})  \chi_{38}(-x_{46})  \chi_{47}(x_{47})  \chi_{55}(x_{48})  \\&&
\chi_{40}(-x_{49})  \chi_{50}(x_{50})  \chi_{58}(x_{51})  \chi_{59}(x_{52})  \chi_{45}(-x_{53})  \chi_{54}(x_{54})  \chi_{48}(-x_{55})  \chi_{56}(x_{56})  \chi_{64}(x_{57})  \chi_{51}(-x_{58})  \chi_{52}(-x_{59})  \chi_{60}(x_{60})  \chi_{61}(x_{61})  \chi_{62}(x_{62})  \chi_{70}(x_{63})  \chi_{57}(-x_{64})  \\&&
\chi_{65}(x_{65})  \chi_{66}(x_{66})  \chi_{67}(x_{67})  \chi_{74}(-x_{68})  \chi_{75}(x_{69})  \chi_{63}(-x_{70})  \chi_{71}(x_{71})  \chi_{72}(x_{72})  \chi_{79}(-x_{73})  \chi_{68}(x_{74})  \chi_{69}(-x_{75})  \chi_{76}(x_{76})  \chi_{77}(x_{77})  \chi_{84}(-x_{78})  \chi_{73}(x_{79})  \chi_{80}(x_{80})  \\&&
\chi_{81}(x_{81})  \chi_{82}(x_{82})  \chi_{88}(-x_{83})  \chi_{78}(x_{84})  \chi_{85}(x_{85})  \chi_{91}(-x_{86})  \chi_{92}(-x_{87})  \chi_{83}(x_{88})  \chi_{89}(x_{89})  \chi_{95}(-x_{90})  \chi_{86}(x_{91})  \chi_{87}(x_{92})  \chi_{93}(x_{93})  \chi_{99}(-x_{94})  \chi_{90}(x_{95})  \chi_{96}(x_{96})  \\&&
\chi_{97}(x_{97})  \chi_{102}(-x_{98})  \chi_{94}(x_{99})  \chi_{100}(x_{100})  \chi_{104}(-x_{101})  \chi_{98}(x_{102})  \chi_{103}(x_{103})  \chi_{101}(x_{104})  \chi_{105}(x_{105})  \chi_{106}(x_{106})  \chi_{107}(x_{107})  \chi_{108}(x_{108})  \chi_{109}(x_{109})  \chi_{110}(x_{110})  \chi_{111}(x_{111})  \\&&
\chi_{112}(x_{112})  \chi_{113}(x_{113})  \chi_{114}(x_{114})  \chi_{115}(x_{115})  \chi_{116}(x_{116})  \chi_{117}(x_{117})  \chi_{119}(-x_{118})  \chi_{118}(x_{119})  \chi_{120}(x_{120})   \\
8 & t_{7}/t_{8}^2 & h_{8}(c_8/r_8)h\cdot \chi_{8}(-(c_8/r_8)^2x_8)\cdot\chi_{1}(x_{1})  \chi_{2}(x_{2})  \chi_{3}(x_{3})  \chi_{4}(x_{4})  \chi_{5}(x_{5})  \chi_{6}(x_{6})  \chi_{7}(-x_{7}x_{8}) \chi_{15}(x_{7})  \cdot \chi_{9}(x_{9})  \chi_{10}(x_{10})  \chi_{11}(x_{11})  \chi_{12}(x_{12})  \chi_{13}(x_{13})  \chi_{22}(x_{14})  \chi_{7}(-x_{15})  \chi_{16}(x_{16}) \\&&
\chi_{17}(x_{17})  \chi_{18}(x_{18})  \chi_{19}(x_{19})  \chi_{20}(x_{20})  \chi_{29}(x_{21})  \chi_{14}(-x_{22})  \chi_{23}(x_{23})  \chi_{24}(x_{24})  \chi_{25}(x_{25})  \chi_{26}(x_{26})  \chi_{27}(x_{27})  \chi_{36}(x_{28})  \chi_{21}(-x_{29})  \chi_{30}(x_{30})  \chi_{31}(x_{31})  \chi_{32}(x_{32})  \\&&
\chi_{33}(x_{33})  \chi_{42}(x_{34})  \chi_{43}(x_{35})  \chi_{28}(-x_{36})  \chi_{37}(x_{37})  \chi_{38}(x_{38})  \chi_{47}(x_{39})  \chi_{40}(x_{40})  \chi_{50}(x_{41})  \chi_{34}(-x_{42})  \chi_{35}(-x_{43})  \chi_{44}(x_{44})  \chi_{45}(x_{45})  \chi_{54}(x_{46})  \chi_{39}(-x_{47})  \chi_{48}(x_{48})  \\&&
\chi_{56}(x_{49})  \chi_{41}(x_{50})  \chi_{51}(x_{51})  \chi_{52}(x_{52})  \chi_{60}(x_{53})  \chi_{46}(-x_{54})  \chi_{62}(x_{55})  \chi_{49}(-x_{56})  \chi_{57}(x_{57})  \chi_{65}(x_{58})  \chi_{67}(x_{59})  \chi_{53}(-x_{60})  \chi_{68}(x_{61})  \chi_{55}(-x_{62})  \chi_{63}(x_{63})  \chi_{72}(x_{64})  \\&&
\chi_{58}(-x_{65})  \chi_{73}(x_{66})  \chi_{59}(-x_{67})  \chi_{61}(-x_{68})  \chi_{69}(x_{69})  \chi_{77}(x_{70})  \chi_{78}(x_{71})  \chi_{64}(-x_{72})  \chi_{66}(-x_{73})  \chi_{74}(x_{74})  \chi_{81}(x_{75})  \chi_{83}(x_{76})  \chi_{70}(-x_{77})  \chi_{71}(-x_{78})  \chi_{79}(x_{79})  \chi_{86}(x_{80})  \\&&
\chi_{75}(-x_{81})  \chi_{87}(x_{82})  \chi_{76}(-x_{83})  \chi_{84}(x_{84})  \chi_{90}(x_{85})  \chi_{80}(-x_{86})  \chi_{82}(-x_{87})  \chi_{88}(x_{88})  \chi_{94}(x_{89})  \chi_{85}(-x_{90})  \chi_{91}(x_{91})  \chi_{92}(x_{92})  \chi_{98}(x_{93})  \chi_{89}(-x_{94})  \chi_{95}(x_{95})  \chi_{96}(x_{96})  \\&&
\chi_{101}(x_{97})  \chi_{93}(-x_{98})  \chi_{99}(x_{99})  \chi_{100}(x_{100})  \chi_{97}(-x_{101})  \chi_{102}(x_{102})  \chi_{103}(x_{103})  \chi_{104}(x_{104})  \chi_{105}(x_{105})  \chi_{106}(x_{106})  \chi_{107}(x_{107})  \chi_{108}(x_{108})  \chi_{109}(x_{109})  \chi_{110}(x_{110})  \chi_{111}(x_{111})\\&&  \chi_{112}(x_{112})  \chi_{113}(x_{113})  \chi_{114}(x_{114})  \chi_{115}(x_{115})  \chi_{116}(x_{116})  \chi_{117}(x_{117})  \chi_{118}(x_{118})  \chi_{120}(-x_{119})  \chi_{119}(x_{120})   \\
\end{array}$
}
\caption{Iwasawa coordinates for $E_8$ with compensating element 
\mbox{$\kappa_i(c_ix_i)=\chi_i(-c_ix_i)h_i(-1/\sqrt{(c_ix_i)^2+1})\wt_i \chi_{i}(-c_ix_i)$: Part 2}}\label{T-last}
\end{table}

\end{landscape}


\section{The $\SL_2$ and $E_7$ cosets}
\subsection{$\SL_2$ symmetry in type IIB supergravity in $D=10$ dimensions}
Type IIB supergravity in $D=10$ spacetime dimensions has  two real scalar fields, the dilation $\phi$ and the axion $\chi$. These can be combined to define a complex scalar field
$$\tau=\chi+ie^{\phi}\;,$$
which parametrizes the upper half-plane 
$\SO(2)\backslash \SL_2(\mathbb{R})\cong \mathbb{H}$ and  transforms under $\SL_2(\mathbb{R})$ 
as the fractional linear transformation
$$\left(\begin{matrix} a & b\\ c & d\end{matrix}\right):\tau\mapsto\dfrac{a\tau+ b}{ c\tau + d}\;.$$
The group $G=\SL_2(\mathbb{R})$ is a global symmetry of the theory and $G$ acts on the asymptotic 
values of the scalar fields at spatial infinity. The orbits of the global symmetry group $G$
yield families of $p$-brane solutions (\cite{CLPS}).

\subsection{Action of $\SL_2(\mathbb{Z})$ on the scalar sector} We will use the results of Section 3 but change some of the notation to follow the conventions of [HPS]. The coset element $V(x)$ 
of the $(\phi, \chi)$ scalar sector is given in Borel gauge by ([HPS])
$$V(x)=\exp\left[\dfrac{\phi(x)}{2}\left(\begin{matrix} -1 & ~0\\ 0 & 1\end{matrix}\right)\right]\exp\left[\chi(x)\left(\begin{matrix} 0 & 1\\ 0 & 0\end{matrix}\right)\right]=\left(\begin{matrix} e^{-\frac{\phi(x)}{2}} & \chi(x)e^{\frac{\phi(x)}{2}}\\ 0 & e^{\frac{\phi(x)}{2}}\end{matrix}\right)\;,$$
where $h=\left(\begin{matrix} -1 & ~0\\ 0 & 1\end{matrix}\right)$ and $e=\left(\begin{matrix} 0 & 1\\ 0 & 0\end{matrix}\right)$ are  $\mathfrak{sl}_2(\mathbb{R})$ generators and $x$ is a spacetime coordinate. To determine the action of $\SL_2(\mathbb{Z})$ on the coset element $V(x)$, it is enough to determine the action of the generators $T$ and $S$ on $V(x)$.

\begin{theorem} The action of the generators  $T$ and $S$ of $\SL_2(\mathbb{Z})$ on $V(x)\in \SO(2)\backslash \SL_2(\mathbb{R})$ is given in Iwasawa coordinates by
$$V(x)\cdot T=\left(\begin{matrix} e^{-{\phi(x)}/{2}} & 0\\ 0 & e^{{\phi(x)}/{2}}\end{matrix}\right)
\left(\begin{matrix} 1 & e^{{\phi(x)}}[e^{-\phi(x)}+\chi(x)]\\ 0 & e^{{\phi(x)}/{2}}\end{matrix}\right)$$
and 
$$V(x)\cdot S=\left(\begin{matrix} e^{{\phi(x)}/{2}} \sqrt{1+\chi(x)^2} & 0 \\ 0 & \dfrac{e^{-{\phi(x)}/{2}}}{\sqrt{1+\chi(x)^2}} \end{matrix}\right)
\left(\begin{matrix} 1 & -\dfrac{e^{-\phi(x)}\chi(x)}{{1+\chi(x)^2}}  \\ 0 & 1\end{matrix}\right)\;.$$
\end{theorem}
\medskip\noindent{\it Proof:} We start with the right action of $T$:
$$V(x)\cdot T=\left(\begin{matrix} e^{-{\phi(x)}/{2}} & \chi(x)e^{{\phi(x)}/{2}}\\ 0 & e^{{\phi(x)}/{2}}\end{matrix}\right)\cdot \left(\begin{matrix} 1 & 1\\ 0 & 1\end{matrix}\right)
=
\left(\begin{matrix} e^{-{\phi(x)}/{2}} & e^{{\phi(x)}/{2}}[e^{-\phi(x)}+\chi(x)]\\ 0 & e^{{\phi(x)}/{2}}\end{matrix}\right),$$
which may be written in Iwasawa coordinates as
$$V(x)\cdot T=\left(\begin{matrix} e^{-{\phi(x)}/{2}} & 0\\ 0 & e^{{\phi(x)}/{2}}\end{matrix}\right)
\left(\begin{matrix} 1 & e^{{\phi(x)}}[e^{-\phi(x)}+\chi(x)]\\ 0 & e^{{\phi(x)}/{2}}\end{matrix}\right)\;.$$
Similarly, we have for the action of $S$:
$$V(x)\cdot S=\left(\begin{matrix} e^{-{\phi(x)}/{2}} & \chi(x)e^{{\phi(x)}/{2}}\\ 0 & e^{{\phi(x)}/{2}}\end{matrix}\right)\cdot \left(\begin{matrix} 0 & 1\\ -1 & 0\end{matrix}\right)= 
\left(\begin{matrix} -tu & t\\ -t^{-1} & 0\end{matrix}\right)\;,$$
where $t=e^{-{\phi(x)}/{2}}$ and $u=\chi(x)e^{\phi(x)}$. However, this last
 expression is not in Iwasawa form, hence we multiply on the left by the compensating element (as in Proposition~\ref{iwasawa})
$$\kappa=\dfrac{1}{\sqrt{t^2u^2+t^{-2}}}\left(\begin{matrix} -tu & -t^{-1}\\ t^{-1} & -tu\end{matrix}\right)$$
to obtain the desired form
$$\kappa \cdot V(x)\cdot S=\left(\begin{matrix} e^{{\phi(x)}/{2}} \sqrt{1+\chi(x)^2} & 0 \\ 0 & \dfrac{e^{-{\phi(x)}/{2}}}{\sqrt{1+\chi(x)^2}} \end{matrix}\right)
\left(\begin{matrix} 1 & -\dfrac{e^{-\phi(x)}\chi(x)}{{1+\chi(x)^2}}  \\ 0 & 1\end{matrix}\right).$$
\endofproof

For example, if $\phi(x)=\chi(x)=0$, that is, at the point $i$, the action reduces to
$$V\cdot T=\left(\begin{matrix} 1 & 0\\ 0 & 1\end{matrix}\right)\left(\begin{matrix} 1 & 1\\ 0 & 1\end{matrix}\right)=\left(\begin{matrix} 1 & 1\\ 0 & 1\end{matrix}\right),$$
$$V\cdot S=\left(\begin{matrix} 1 & 0\\ 0 & 1\end{matrix}\right)\left(\begin{matrix} 1 & 0\\ 0 & 1\end{matrix}\right)=\left(\begin{matrix} 1 & 0\\ 0 & 1\end{matrix}\right).$$
Note that the latter result is expected since $S\in K=\SO(2)$, so $S$  fixes $V$ at the point $i$.

\subsection{The action of $G(\mathbb{Z})$ on $K\backslash G(\mathbb{R})$}
\label{Sec E7}
 Let $G=G(\R)$ be a split, simply connected Lie group with maximal compact subgroup $K$ as in Section~\ref{Sec Iwasawa}. Let $\mathfrak{g}$ be the Lie algebra of $G$.
Let $V$ be a highest-weight representation of $\mathfrak{g}$. Then $G\leq \Aut(V)$. Provided $V$ contains an admissible lattice $V_\Z$,  we can define a $\Z$-form $G(\Z)$ of $G$ as the stabilizer
 of $V_\Z$ (\cite{St}).
 In order to describe the action of $G(\Z)$ on $K\backslash G(\R)$, we need a generating set for $G(\Z)$.
 The following result provides generators for a large class of integral forms:
 \begin{theorem}[\cite{CCa},\cite{BC}]
 Let $V$ be a highest-weight representation of $\mathfrak{g}$ and suppose that the lattice generated by the weights of $V$ contains all the fundamental weights.
 Then $G(\Z)$ is generated by 
$\chi_{\alpha_i}(1)$ and $\widetilde{w}_{\alpha_i}=\chi_{\alpha_i}(1)\chi_{-\alpha_i}(-1)\chi_{\alpha_i}(1)$, for $i=1,\dots , n$.
\end{theorem}

The Iwasawa coordinates are 
$$Khu=K\prod_{j=1}^n h_{j}(t_{j})\prod_{r=1}^{N} \chi_r(x_r)$$
where $t_j, x_r\in\R,$ and $t_j>0$.

The right action of $\chi_{i}(1)$ for $i=1,\dots ,n$, is given by
\begin{eqnarray*}
&&\prod_{j=1}^n h_{j}(t_j)\prod_{r=1}^{N} \chi_{r}(x_r)\cdot \chi_{i}(1),
\end{eqnarray*}
which can be rewritten as an expression   in  Iwasawa form
 using the collection algorithms of \cite{CMH}.
However the right action of $\widetilde{w}_{i}$ for $i=1,\dots ,n$, is given by
\begin{eqnarray*}
&&\prod_{j=1}^n h_{j}(t_j)\prod_{r=1}^{N} \chi_{r}(x_r)\cdot \widetilde{w}_{i},
\end{eqnarray*}
which is not   in Iwasawa form. 
 To remedy this, we use the  compensating element $\kappa\in K$ of Section~\ref{Sec Iwasawa} to  left multiply  
the expression 
by $\kappa_i(c_ix_i)$ 
which gives an expression in Iwasawa form.

This result applies to  the
 noncompact split real form $E_{7(+7)}(\mathbb{R})$ of the exceptional Lie group  $E_7$ with the following form of $E_7(\mathbb{Z})$:
$$E_{7(+7)}(\mathbb{Z}) = E_{7(+7)}(\mathbb{R}) \cap \Sp_{56}(\mathbb{Z})\;,$$
discovered in \cite{HT} following \cite{CJ} in the framework of type II string theory. Soul\'e gave a rigorous mathematical proof that the $E_{7(+7)}(\mathbb{Z})$ of Hull and Townsend coincides with the Chevalley $\Z$-form $G(\Z)$ of $G=E_7$ ([S]). Here $E_{7(+7)}(\mathbb{R}) \cap \Sp_{56}(\mathbb{Z})$ is the stabilizer of the standard lattice in the fundamental representation of $E_7$ which has dimension 56. The charge lattice of \cite{HT} can be normalized to coincide with the lattice $V_{\Z}$. Once a basis for $V_{\Z}$ has been chosen, the $E_{7}(\mathbb{Z})$ orbits can be computed explicitly in terms of this basis. 
The  non-compact split real form $E_{7(+7)}(\mathbb{R})$ has maximal compact subgroup $K=[\SU(8)/\{\pm \Id\}]$. The coset $[\SU(8)/\{\pm \Id\}]\backslash E_{7(+7)}(\mathbb{R})$ occurs as a scalar coset for dimensional reduction of $N=1$ supergravity in 11 dimensions  to $N = 8$ supergravity in four dimensions (\cite{CJ}).

\section{The action of $\SL_2(\Z)$ on the charge lattice}\label{charge}
\label{Sec SL2 action}

For now,  we assume the existence of dyons, that is, particles that carry both electric and magnetic charge,
see \cite{O1} and \cite{O2}. 
The Dirac quantization condition  was generalized  by Zwanziger and Schwinger to a pair of dyons with 
$$\text{electric charges $q_1$ and $q_2$}$$
$$\text{magnetic charges $g_1$ and $g_2$}$$
such that 
\(
q_1g_2-q_2g_1=2\pi n\hbar, ~~\ n=0,\ \pm 1,\ \pm 2,\cdots .
\label{ZS cond} 
\)
Witten solved this quantization condition giving a general family of solutions $\{(q,g)\}$ of electric and magnetic charges respectively, with
$$q+{i}g=q_0(m\tau+n)$$
where $m,n\in \mathbb{Z}$ are coprime, $q_0$ is the electron charge, and $\tau$ is the parameter 
$$\tau=\dfrac{\theta}{2\pi}+\dfrac{2\pi i n\hbar}{q_0^2}\;,$$
where 
$\theta$ is the ``vacuum angle"
([Wi]). Here $\tau$ is a complex variable 
with positive imaginary part and which 
depends on dimensionless parameters corresponding to the particular theory. 

If we equate real and imaginary parts of $q+{i}g=q_0(m\tau+n)$ and set $m=0$, then we obtain
$q=q_0n,\ n\in\Z$.
This means that the electric charge is an integer multiple of the basic electron charge $q_0$.

The collection of possible dyonic charges is encoded in the 
{\it charge lattice} $$\mathcal{Q}=\{q+{i} g\mid q,g~{\rm satisfy~condition~} \eqref{ZS cond} \}\subset \C\;,$$
depicted in Figure~\ref{ChargeLattice}.

\begin{figure}[h]
	\centering
		\includegraphics[scale=0.21]{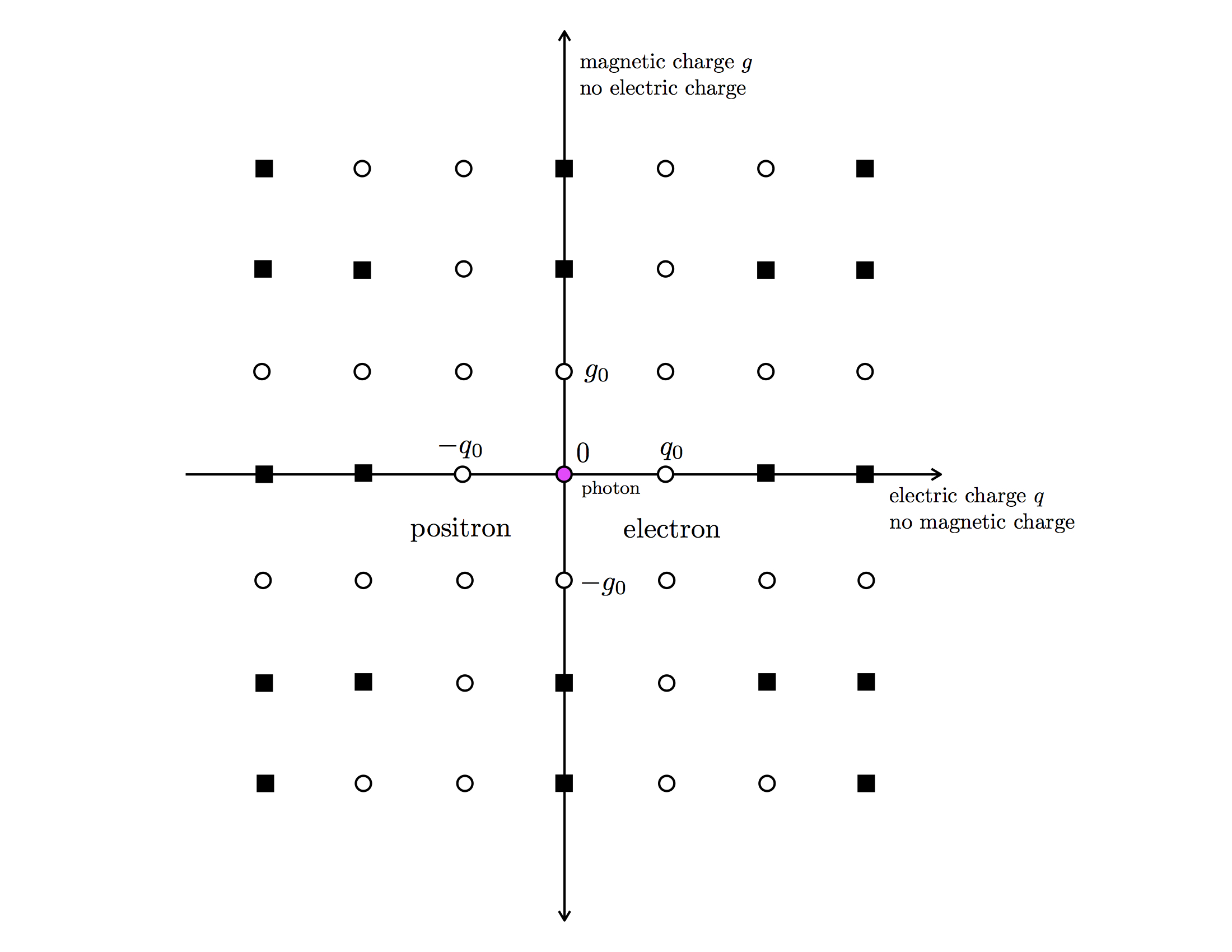}
		\caption{The charge lattice $\mathcal{Q}$}
	\label{ChargeLattice}
\end{figure}

Here electric particles lie on the real axis,  monopoles lie on the imaginary axis,
and a general point in the $(q,g)$ plane off the axes is a  dyon. The origin represents 
a state with no electric or magnetic charge, hence could be taken to represent a photon
or a charge-less (flat) background.  
For magnetically neutral states, electric charges are integral multiplies of $q_0$.
The open circles represent `primitive vectors' on the lattice. These are vectors that can be connected to the origin by a line that does not intersect any other lattice points. Equivalently, $(m,n)$ are relatively prime for these points in the equation $q+{i}g=q_0(m\tau+n)$. The action of the modular group $SL_2(\Z)$ preserves primitive vectors in the charge lattice. 

Given a specific theory, an important question is identifying 
 the subset of the charge lattice that can be realized by single-particle states, as
opposed to multi-particle states. Here it is the primitive vectors that represent single-particle states (\cite{O2}).

We now look at symmetries of the charge lattice. Firstly, the Weyl group $W\cong \mathbb{Z}_2$ permutes the axes of the charge vector lattice, interchanging electric and magnetic charges.
Secondly, the generalization of duality symmetry of Maxwell is a rotation by $\pi/2$ which exchanges electrically charged particles and magnetic monopoles. The $\pi/2$-rotational symmetry  
$$(q_0,0)\to (0,-g_0)\;, \qquad
(-q_0,0)\to (0,g_0)\;,$$
$$(0,g_0)\to (q_0,0)\;, \qquad  \quad
(0,-g_0)\to (-q_0,0)$$
preserves the magnitudes of electric and magnetic charges.
The rotation group $R$ of the elementary charges is hence the cyclic group of order 4 generated by the matrix
$$a=\left(
\begin{array}{cc}
 ~0 & 1 \\
 -1 & 0 \\
\end{array}
\right)\;,$$
that is
$$R=\SL_2(\Z)\cap \SO(2).$$
The discrete subgroup $R$  acts by rotating the elementary charge vector $(q_0,0)$ by integral multiples of $\pi/2$,  while keeping its length fixed. 
The group $R$ is the full rotation group of the charge lattice, and the rotational symmetry preserves primitive vectors and interchanges electric and magnetic charges.

%
 
 The charge lattice can be characterized algebraically. 
Let  $G=\SL_2(\mathbb{C})$ and consider the vector spaces
$$V_{\Z}=\Z\left[\begin{matrix} 1 \\ 0\end{matrix}\right]\oplus \Z\left[\begin{matrix} 0 \\ 1\end{matrix}\right],$$
$$V_{\C}=\C\left[\begin{matrix} 1 \\ 0\end{matrix}\right]\oplus \C\left[\begin{matrix} 0 \\ 1\end{matrix}\right].$$
Then $V_{\C}$ is a highest weight module for the Lie algebra $\mathfrak{g}=\mathfrak{sl}_2(\C)$. 
The vector space $V_{\C}$ can be identified with the charge lattice $\mathcal{Q}$ by setting the magnitude of the elementary electric charge $q_0$ equal to 1. The group $\SL_2(\Z)$ can be viewed as  
the subgroup of  $\SL_2(\C)$ that stabilizes the vector space $V_{\Z}$.

In the absence of magnetic charges, the space of states allowed by the Dirac quantization condition, $p=p_0n,\ n\in\Z$, 
is the set of points  $(p,q)$ with $p,q\in\Z-\{0\}$. 
In general, the charge lattice contains pairs of charges $(p,q)$ 
on both the real and imaginary axes with $p,q\neq 0$. As before, primitive vectors correspond to coordinates
 $(p,q)$ with $p,q\in\Z-\{0\}$, $gcd(p,q)=1$. 

The following proposition is easy to verify.

\begin{proposition} The action of $\SL_2(\Z)$ on the set of allowed states $(p,q)$ with $p,q\in\Z-\{0\}$ 
is not transitive. The action of the generators of $\SL_2(\Z)$ is:

$$\left(\begin{matrix}  
1 & 1 \\
 0 & 1  \end{matrix}\right)\left[\begin{matrix} 1 \\ 0 \end{matrix}\right]= 
\left[\begin{matrix} 1 \\ 0\end{matrix}\right],\quad \left(\begin{matrix}  
1 & 1 \\
 0 & 1  \end{matrix}\right)\left[\begin{matrix} 0 \\ 1 \end{matrix}\right]= 
\left[\begin{matrix} 1 \\ 1\end{matrix}\right]$$
$$\left(\begin{matrix}  
1 & 0 \\
 1 & 1  \end{matrix}\right)\left[\begin{matrix} 1 \\ 0 \end{matrix}\right]= 
\left[\begin{matrix} 1 \\ 1\end{matrix}\right],\quad \left(\begin{matrix}  
1 & 0 \\
 1 & 1  \end{matrix}\right)\left[\begin{matrix} 0 \\ 1 \end{matrix}\right]= 
\left[\begin{matrix} 0 \\ 1\end{matrix}\right].$$
Thus there is a subset of fixed vectors.
The action of  $\SL_2(\Z)$ on the set  $(p,q)$ with $p,q\in\Z-\{0\}$, $gcd(p,q)=1$ is  transitive.

\end{proposition}

\subsection{Spectrum-generating symmetries for BPS solitons}

The concept of dyons and their corresponding charges from electromagnetism 
generalize to `higher notions' of particles, meaning extended objects, for example,  $p$-branes
(see \cite{DGHT}). 
The case $p=0$ corresponds to the case of point particles described above. 
A {\it  BPS-saturated soliton}  is a  solution of a supergravity theory that describes the low energy limit of a string theory in which infinite $p$-branes occupy a longitudinal submanifold of spacetime. Associated to such a solution is a vector space of electric or magnetic charges, or both, if dyonic $p$-branes are present.

One important goal is to determine the spectrum-generating symmetries for the fundamental BPS solitons. In~\cite{CLPS}, it was shown that the standard global supergravity symmetry group is not sufficient for the purpose of generating complete sets of $p$-brane solitons. The additional ingredient needed is a  scaling transformation that allows one to map between BPS solitons with different masses. 
To that end, consider the matrices
$$\sigma(s)=\left(\begin{matrix} s & 0\\ 0 & 1\end{matrix}\right)\qquad  \text { and } \qquad \tau(t)=\left(\begin{matrix} 1 & 0\\ 0 & t\end{matrix}\right)\;,$$
for fixed $s,t\in \mathbb{R}_{>0}$. These elements $\sigma$ and $\tau$ belong to $\GL_2(\mathbb{R})$ and 
are the `trombone' scaling symmetries of [CLPS]. In particular,  $\sigma$ and $\tau$ each contain only single $\mathbb{R}_{>0}$ parameters.
They are symmetries of the equations of motion, but not of the action, are  invariant under dimensional reduction,
 and preserve the dilation and axion scalar fields (\cite{CLPS}).
Indeed, for any $s,t\in \R_{>0}$, we have
$$\left(\begin{matrix} s & 0\\ 0 & 1\end{matrix}\right)\left(\begin{matrix} q_0 \\ g_0 \end{matrix}\right)=\left(\begin{matrix} sq_0 \\ g_0 \end{matrix}\right),$$
$$\left(\begin{matrix} 1 & 0\\ 0 & t\end{matrix}\right)\left(\begin{matrix} q_0 \\ g_0 \end{matrix}\right)=\left(\begin{matrix} q_0 \\ tg_0 \end{matrix}\right),$$
so $\sigma(s)$ rescales the elementary electric charge and $\tau(t)$ rescales the elementary magnetic charge.


Let $P(\R)$ be the subgroup of $\GL_2(\mathbb{R})$ generated by $K=\SO(2)$ and the elements 
$$\sigma(\R)=\left \{\left(\begin{matrix} s & 0\\ 0 & 1\end{matrix}\right)\mid s\in\R\right \}$$ and 
$$\tau(\R)=\left \{\left(\begin{matrix} 1 & 0\\ 0 & t\end{matrix}\right)\mid t\in\R\right \}.$$

As in \cite{CLPS}, we have the following.
\begin{theorem}  At the classical level, the subgroup $P(\R)=\langle \SO(2),\ \sigma(\R),\  \tau(\R)\rangle$ of $\GL_2(\mathbb{R})$ is the spectrum-generating symmetry group  for the parameter space of BPS solitons that preserve half the supersymmetry
 of Type IIB supergravity in $D=10$ dimensions. 
 \end{theorem}

Next consider 
$$P(\Z)=\langle \SL_2(\Z),\ \sigma(\Z),\  \tau(\Z)\rangle $$
where
$$\sigma(\Z)=\left \{\left(\begin{matrix} s & 0\\ 0 & 1\end{matrix}\right),\  s\in \mathbb{Z} \right \}\quad \text{ and } 
\quad \tau(\Z)=\left \{\left(\begin{matrix} 1 & 0\\ 0 & t\end{matrix}\right),\ t\in \mathbb{Z}\right \}.$$
Then $\SL_2(\Z)$ is naturally contained in $P(\Z)$. However, $P(\Z)$ is not a group, as the inverses of $\sigma(\Z)$ and $\tau(\Z)$ are not contained in $P(\Z)$. We can identify $P(\Z)$ as a {\it monoid} 
 and in fact a submonoid of $\GL_2(\Z_{\geq 0})$, where $\GL_2(\mathbb{Z}_{\geq 0})=M_2(\mathbb{Z}_{\geq 0})\cap GL_2(\Z)$ is  defined as the monoid of matrices of determinant $\neq 0$ with entries in $\mathbb{Z}_{\geq 0}$. The elements of $\GL_2(\mathbb{Z}_{\geq 0})$ are not necessarily invertible, since inverses may not exist over $\Z_{\geq 0}$.

To see that $P(\Z)$ is  a submonoid of $\GL_2(\Z_{\geq 0})$, we use the following.

\begin{theorem}\label{invert}([R])
\begin{enumerate}
\item $\SL_2(\Z_{\geq 0})=M_2(\mathbb{Z}_{\geq 0})\cap \SL_2(\Z)$ is a free monoid on two generators 
$\left(\begin{smallmatrix} 
1 & 1\\
0 & 1
\end{smallmatrix}\right)$,
$\left(\begin{smallmatrix} 
1 & 0\\
1 & 1
\end{smallmatrix}\right)$.

\medskip\item $\SL_2(\Z)$ can be decomposed using 8 copies of $\SL_2(\Z_{\geq 0})$:

\bigskip
$\SL_2(\Z)=\SL_2(\Z_{\geq 0})\ \sqcup\  \gamma \SL_2(\Z_{\geq 0})\ \sqcup\  \SL_2(\Z_{\geq 0})  \gamma \ \sqcup\  \gamma \SL_2(\Z_{\geq 0}) \gamma$

\medskip
\qquad $
\ \sqcup\  (-1)\SL_2(\Z_{\geq 0})\ \sqcup\   (-1)\gamma \SL_2(\Z_{\geq 0})\ \sqcup\   (-1)\SL_2(\Z_{\geq 0})   \gamma\ \sqcup\   (-1)\gamma \SL_2(\Z_{\geq 0}) \gamma,
$

\bigskip
where $\gamma=\left(\begin{smallmatrix} 
~0 & 1\\
-1 & 0
\end{smallmatrix}\right)$ and $(-1)=\left(\begin{smallmatrix} 
-1 & ~0\\
~0 & -1
\end{smallmatrix}\right)$

\medskip\item $\GL_2(\Z)$ can be decomposed using 2 copies of $\SL_2(\Z)$:\ $\GL_2(\Z)=\SL_2(\Z)\ \sqcup\  \left(\begin{smallmatrix} 
1 & ~0\\
0 & -1
\end{smallmatrix}\right) \SL_2(\Z)$.
\end{enumerate}

\end{theorem}

It follows that the generators for $\GL_2(\mathbb{Z}_{\geq 0})$ are 
$$\left(\begin{matrix} 1 & u\\ 0 & 1\end{matrix}\right)\text { and }\left(\begin{matrix} 1 & 0\\ v & 1\end{matrix}\right)$$
for $u,v\in{\mathbb{Z}_{\geq 0}}$ and
$$\left(\begin{matrix} s & 0\\ 0 & 1\end{matrix}\right)\text { and } \left(\begin{matrix} 1 & 0\\ 0 & t\end{matrix}\right)$$
for $s,t\in \mathbb{Z}_{> 0}$. Thus $P(\Z)$ is a submonoid of $\GL_2(\Z_{\geq 0})$.

Let $(p_0,0)^T$, $(0,q_0)^T$ be the elementary vectors of electric charge. It is convenient to choose $p_0=1$. Then we must have $q_0\neq 1$ since the vector $(p_0,q_0)$ is primitive. We can generate all primitive vectors using the $P(\Z)$ action on $(p_0,0)^T$ and $(0,q_0)^T$ in the following way. We choose matrices 
$$m_1=\left(\begin{matrix} n & b\\ q_0 & d\end{matrix}\right),\ m_2=\left(\begin{matrix} a & p_0\\ c & m\end{matrix}\right),\ m_3=\left(\begin{matrix} s & 0\\ 0 & t\end{matrix}\right)\in P(\Z)$$
with $a$, $b$, $c$, $d$, $p_0$, $q_0$, $m$, $n$, $s$, $t\in\Z$, in such a way that 
$$\gcd(n,q_0)=\gcd(p_0,m)=\gcd(s,t)=1.$$ Then
\begin{align*}
\left(\begin{matrix} n & b\\ q_0 & d\end{matrix}\right)  \left(\begin{matrix} 1 \\ 0 \end{matrix}\right)=
\left(\begin{matrix} n \\ q_0 \end{matrix}\right),\quad
\left(\begin{matrix} a & p_0\\ c & m\end{matrix}\right)  \left(\begin{matrix} 0 \\ 1 \end{matrix}\right)=
\left(\begin{matrix} p_0 \\ m \end{matrix}\right),\quad\text{and}\quad
\left(\begin{matrix} s & 0\\ 0 & t\end{matrix}\right)\left(\begin{matrix} p_0 \\ q_0 \end{matrix}\right)=\left(\begin{matrix} sp_0 \\ tq_0 \end{matrix}\right).
\end{align*}
Thus all the primitive vectors are generated by the action of $m_1$, $m_2$ or $m_3\in P(\Z)$.

 Even though the matrices $m_1$, $m_2$ and $m_3$ have integral entries and are contained in $P(\Z)$, their inverses cannot be contained in $P(\Z)$. This it is a subtle question to determine what it means to `invert' a charge under the action of $P(\Z)$.  One can always invert charges under the action of the group $\SL_2(\Z)$ which preserves the primitive vectors and and hence maps between allowed charge vectors. Hence the group  $\SL_2(\Z)$ generates the spectrum of physically distinct states in a single fixed vacuum. 

Theorem~\ref{invert} also  ensures that charges can be inverted under the action of the group $\GL_2(\Z)$.

However, we have a quantum anomaly of the theory: the trombone symmetries $\sigma(\Z)$ and $\tau(\Z)$ cannot be inverted over $\Z$ in the monoid $P(\Z)$ that they generate.

We may summarize the above discussion in the following.
\begin{lemma} Elements of the cosets
$$\SO(2)\sigma(\Z),\quad \SO(2)\tau(\Z)$$
generate the full charge lattice from the elementary charge vector $(1,0)^T$ and preserve the moduli of the scalar fields.
\end{lemma}
The elements of $\{\SO(2)\sigma(\Z)$, $\SO(2)\tau(\Z)\}$ are not integral matrices and moreover do not generate a group, but rather  cosets of $\SO(2)\backslash \GL_2(\R)$.

This shows that the usual duality group symmetries at the classical level are not inherited in the same form at the quantum level and one must take care when describing the discrete symmetries in the quantum theory.

\section{U-duality groups in general}\label{other}

 Our method for computing Iwasawa coordinates of the quotient $K\,\backslash G$ and the action of the $\mathbb{Z}$-form $G(\mathbb{Z})$ on the coset $K\,\backslash G$  do not depend on the choice of group $G$. Our techniques can hence be applied to other U-duality groups, including the hyperbolic Kac--Moody group $E_{10}$, which is conjectured to be the U-duality group of 11-dimensional supergravity in 1 dimension (\cite{HT}, \cite{DHN}).

In \cite{AC}, the authors obtained a finite presentation of an important class of hyperbolic Kac-Moody groups, including $E_{10}(\R)$ and $E_{10}(\Z)$. Hence the action of $E_{10}(\Z)$ on the coset $K(E_{10}(\R))\backslash E_{10}(\R)$ is given by finitely many rules which can be determined by a generalization of the methods given here.

 The $\Z$-form $G(\Z)$, for any simple and simply connected Chevalley group or Kac--Moody group $G$, may be defined as the stabilizer 
$$G({\mathbb{Z}})=\{g\in G({\mathbb{R}})\mid g\cdot V_{\Z}\subseteq V_{\Z}\}$$
of a lattice $V_{\Z}$ in a suitable highest weight module $V$ for the Lie algebra $\mathfrak{g}$ of $G$ (\cite{St}, \cite{BC}). In many cases of physical interest, a fundamental representation serves as a suitable choice of $V$.

For a general discrete duality group $G(\Z)$, we may also determine the  action of $G(\Z)$  on $\Z$-forms $V_{\Z}$  of fundamental modules $V$ for Lie algebras and Kac--Moody algebras in general. The charge lattice $\mathcal{Q}$ of Section~\ref{charge} can be normalized to coincide with the lattice $V_{\Z}$. Once a basis for $V_{\Z}$ has been chosen, the $G(\mathbb{Z})$ orbits on the charge lattice can be computed explicitly in terms of this basis. We hope to take  up these open questions in forthcoming work.

\end{document}